\documentclass[12pt]{article}

\usepackage[T1]{fontenc}
\usepackage{graphicx}
\usepackage[latin9]{inputenc}
\usepackage{amsmath,mathrsfs,mathtools,amsthm}
\usepackage{amsmath,amsfonts,amssymb,epsfig}
\usepackage[usenames,dvipsnames]{color}

\usepackage{float}
\usepackage{array}

\newtheorem{defi}{Definition}
\newtheorem{lem}{Lemma}
\newtheorem{thm}{Theorem}
\newtheorem{prop}{Proposition}

\newtheorem{rem}{Remark}

\newtheorem{ex}{Example}

\newcommand{\C}{\mbox{$\cal C$}}

\newcommand{\ben}{\begin{equation*}}
\newcommand{\een}{\end{equation*}}

\newcommand{\comment}[1]{}

\newcommand{\F}{\mathbb{F}}

\newcommand{\x}{\mbox{\bf x}}

\begin{document}

\author{ Minjia Shi$^\ast$\thanks{smjwcl.good@163.com},
	Xuan Wang \thanks{wang\_xuan\_ah@163.com}, Junmin An\thanks{junmin0518@sogang.ac.kr}, Jon-Lark Kim\thanks{jlkim@sogang.ac.kr},
	\thanks{Minjia Shi and Xuan Wang are with the Key Laboratory of Intelligent Computing Signal 		Processing, Ministry of Education, School of Mathematical Sciences, Anhui
		University, Hefei 230601, China; State Key Laboratory of integrated Service Networks, Xidian University, Xi'an,
		710071, China. Junmin An and Jon-Lark Kim are with Department of Mathematics, Sogang University, Seoul, South Korea. This research (M. Shi) is supported by National Natural Science Foundation of China (12471490). This research (J.-L. Kim and J. An) is partially supported by the 4th BK 21 ``Nurturing team for creative and convergent mathematical science talents'' of Department of Math at Sogang University.}}

\title{Log-Concave Sequences in Coding Theory}


\date{}

\maketitle
	\begin{abstract}
	We introduce the notion of logarithmically concave (or log-concave) sequences in Coding Theory. A sequence $a_0, a_1, \dots, a_n$ of real numbers is called log-concave if $a_i^2 \ge a_{i-1}a_{i+1}$ for all $1 \le i \le n-1$. A natural sequence of positive numbers in coding theory is the weight distribution of a linear code consisting of the nonzero values among $A_i$'s where $A_i$ denotes the number of codewords of weight $i$. We call a linear code log-concave if its nonzero weight distribution is log-concave. Our main contribution is to show that all binary general Hamming codes of length $2^r -1$ ($r=3$ or $r \ge 5$), the binary extended Hamming codes of length $2^r ~(r \ge 3)$, and the second order Reed-Muller codes $R(2, m)~ (m \ge 2)$ are all log-concave  while the homogeneous and projective second order Reed-Muller codes are either log-concave, or 1-gap log-concave.
Furthermore, we show that any MDS $[n, k]$ code over $\mathbb F_q$ satisfying $3 \leqslant k \leqslant n/2 +3$ is log-concave if $q \geqslant q_0(n, k)$ which is the larger
root of a quadratic polynomial. Hence, we expect that the concept of log-concavity in coding theory will stimulate many interesting problems.

	\end{abstract}

\section{Introduction}

Many interesting sequences in algebra, combinatorics, geometry, analysis, probability and statistics have been known to be log-concave or unimodal. We recall that a sequence $a_0, a_1, \dots, a_n$ of real numbers is called {\em log-concave} if $a_i^2 \geqslant a_{i-1}a_{i+1}$ for all $1 \leqslant i \leqslant n-1$. A sequence $a_0, a_1, \dots, a_n$ of real numbers is called {\em unimodal} if there is an index $0 \leqslant j \leqslant n$ such that $a_0 \leqslant \cdots  \leqslant a_{j-1} \leqslant a_j \geqslant  a_{j+1} \geqslant \cdots  \geqslant a_n$. For instance, the well-known binomial coefficient sequence ${n \choose i}$ $(0 \leqslant i \leqslant n)$ is log-concave and unimodal. We refer readers to the survey papers by Stanley~\cite{Sta} and Brenti~\cite{Bre}.

Some recent mathematical breakthroughs on log-concavity occurred in matroid theory~\cite{Wel}. Matroid theory shares common concepts in graph theory and linear algebra. More precisely, a matroid $M$ is a pair $(E, I)$, where $E$ is the finite set of elements of the matroid, and $I$ a family of subsets of $E$ called the independent sets of the matroid, satisfying (i) any subset of an independent set is independent and (ii) if $I_1$ and $I_2$ are independent and $|I_1| < |I_2|$, then there exists $e \in  I_2 -I_1$ such
that $I_1 \cup \{e \}$ is independent.  The Tutte polynomial of a matroid is a generalization of the chromatic polynomial of a graph.
The {\em Tutte polynomial} $T(M; x, y)$ of a matroid $M$ is defined as
\[
T(M; x, y) = \sum_{A \subset E} {(x-1)^{{\mbox{rk}}(E)-{\mbox{rk}}(A)}} {(y-1)^{|A|-{\mbox{rk}}(A)}}.
\]
A special case of the Tutte polynomial is the {\em characteristic polynomial}
\[
\chi(M, q) = (-1)^{{\mbox{rk}}(E)} T(M; 1-q, 0)=\sum_{A \subset E} (-1)^{|A|} q^{r-{\mbox{rk}}(A)}.
\]
June Huh, et al. proved the long-standing conjecture that the sequence of absolute values of coefficients of the characteristic polynomial $\chi(M, q)$ of a matroid is unimodal and log-concave \cite{AKH2015}, \cite{AKH}, \cite{KH}. There are two standard examples of matroids: the cycle matroid of a graph and a vector matroid defined by a matrix (this might be the origin of term matroid). Vector matroids occurs naturally in error-correcting codes as seen below.

 A vector matroid is constructed from a vector space $V=\F_q^m$. Let ${\bf{v}}_1, {\bf{v}}_2, \dots, {\bf{v}}_n$ in $V$. We take $E = \{1, 2, \dots, n \}$ and the subset $I$ of $E$ is independent if and only if $\{{\bf{v}}_i ~|~ i \in I \}$ is linearly independent in $V$. Then we can represent the elements of the vector matroid as the columns of an $m \times n$ matrix over $\F_q$ with independence as linearly independence of the columns of the matrix.
A linear $[n, k]$ code $\C$ can be described in terms of a generator matrix or a parity check matrix. Hence a linear code is related to a special vector matroid.
 The Tutte polynomial $T(M; x, y)$ of that vector matroid has a connection with the weight enumerator $W_{\C} (x, y)$ in two variables of the linear $[n, k]$ code $\C$~\cite{Cam} it comes from:
\begin{equation}\label{eq-wt-Tut}
  W_{\mathcal C}(x,y)=y^{n-{\mbox{dim}}(\C)} (x-y)^{{\mbox{dim}}(\C)} T \left( M; \frac{x+(q-1)y}{x-y}, \frac{x}{y} \right).
\end{equation}

It is easy to see that when a linear code has only even weight, its weight distribution is not log-concave in general. For example, the binary $[n, n-1, 2]$ even code has 1, 0, $\binom{n}{2}, \dots$ as the first three vectors and so the sequence is not log-concave.
 In fact, since the characteristic polynomial of a connected graph~\cite[p.65]{Rea} has nonzero coefficients except for the constant, it is natural to consider the nonzero weight distribution of a linear code. However, although the characteristic polynomial of a matroid is log-concave in general, the nonzero weight distribution of a linear codes is not always log-concave. Hence the above equation (\ref{eq-wt-Tut}) does not seem to give us proper information whether the nonzero weight distribution of a linear code is log-concave or not.
 Therefore, it appears that the determination of a log-concavity of the nonzero weight distribution of a random linear code is a very difficult problem since it is an NP-hard problem to find the exact minimum weight of a random linear code~\cite{Var}.

 Currently, the only method to determine whether a given nonzero weight distribution of a linear code is log-concave is to check actually the log-concavity condition of the definition. In general, it is complicated to check the log-concavity because the weights tend to involve several parameters and complex summations. To our surprise, after using various estimation techniques, we were able to show that several families of well-known linear codes have log-concave nonzero weight distributions.

The main purpose of the paper is to show that all binary general Hamming codes of length $2^r -1$ ($r=3$ or $r \ge 5$), the extended Hamming codes of length $2^r ~(r \ge 3)$, the second order Reed-Muller codes $R(2, m)~ (m \ge 2)$ are log-concave while the homogeneous and projective second order Reed-Muller codes are log-concave or 1-gap log-concave. Furthermore, we show that MDS codes of moderate lengths and dimensions are log-concave. Since we considered a few families of linear codes, we expect that the log-concavity topic in coding theory will stimulate many research problems in the future and will be a bridge between combinatorics and coding theory.

The paper consists of six sections. In Section 2, we recall some basic notions and definitions from coding theory. In Section 3, we show that binary Hamming codes except for one length and the binary extended Hamming codes are log-concave. Section 3 proves that the second order Reed-Muller codes are log-concave and that homogeneous and projective second order $q$-ary Reed Muller codes
are 1-gap log-concave. Section 5 shows that MDS codes of dimension at most $n/2 +3$ under some condition are log-concave. In Section 6, we conclude with a summary of obtained results, and offer some challenging open problems.



\section{Preliminaries}

We present some facts from coding theory. We refer to~\cite{GurRudSud},~\cite{HufKimSol},~\cite{HuffmanPless},~\cite{JoyKim},~\cite{MacSlo} for more details.

Let $\mathbb F_q$ be a finite field with $q$ elements, where $q=p^r$ for some prime $p$ and a positive integer $r \ge 1$. We denote the $n$-dimensional vector space over $\F_q$ by
$$\F_q^n = \{\x=(x_1, x_2, \dots, x_n) ~|~ x_i \in \F_q {\mbox{ for all }} i \}.$$

\begin{defi}{\em
A {\em linear $[n, k]$ code} (shortly {\em $[n, k]$ code}) $\mathcal C$ of length $n$ with dimension $k$ over $\F_q$ is a $k$-dimensional subspace of $\F_q^n$.
 If $q=2$, $\mathcal C$ is called a {\em binary} code.
The {\em weight} of ${\bf{x}}=(x_1, \dots, x_n) \in \F_q^n$ is the
number of nonzero coordinates, denoted by wt({\bf{x}}).  The {\em Hamming distance} $d({\bf{x}},{\bf{y}})$ between ${\bf{x}},{\bf{y}} \in \F_q^n$ is wt$({\bf{x}}-{\bf{y}})$.
The {\em minimum distance (weight) $d$} of a linear code
$\C$ is the minimum of wt({\bf {x}}), ${\bf{0 \ne x}} \in \C$. We denote it by an {\em $[n,k,d]$ code} over $\mathbb \F_p$.
A {\em Maximum Distance Separable(MDS) code} is an $[n, k, d]$ code satisfying  $d=n-k+1$.
}
\end{defi}

\begin{defi}{\em
A {\em generator matrix} for an $[n, k]$ code $\C$ over $\F_q$ is a
$k \times n$ matrix $G$ whose rows form a basis for $\C$.
A generator matrix of the form $[I_k | A]$ where $I_k$ is the $k \times k$ identity matrix is called in {\em standard form}.
Given an $[n, k]$ code $\C$ over $\F_q$, there is an $(n-k)\times n$ matrix $H$, called a {\em parity check matrix} for $\C$, defined
by
\[
\C = \{ \x \in  \F_q^n ~|~ H{\x}^T = {\bf 0} \}.
\]
}
\end{defi}

\begin{defi}{\em
The {\em dual} of $\C$ is $\C^{\perp}=\{\x \in \mathbb F_2^n~|~ \x \cdot {\bf {c}} = 0 {\mbox{ for any }} {\bf {c}} \in \C\}$, where the dot product is the usual inner product. A linear code $\C$ is called {\em self-dual} if $\C=\C^{\perp}$ and {\em self-orthogonal} if $\C \subseteq \C^{\perp}$.
}
\end{defi}

\begin{defi}{\em
Given an $[n, k, d]$ code $\C$ over $\F_q$, define $A_i$ as the number of codewords of $\C$ whose weight is $i$. The sequence $A_0, A_1, \dots, A_n$ is called the weight distribution of $\C$, denoted by $A(\C)$. The {\em weight enumerator} $A_{\C}$ is defined by
\[
A_{\C} (x, y) = \sum_{i=0}^{n} A_i x^{n-i}y^i = x^n + A_dx^{n-d}y^d + \cdots + A_ny^n.
\]
}
\end{defi}

Any $[n, k, d]$ code $\C$ has $A_0 =1$, $A_1 = A_2 = \dots = A_{d-1}=0$, and $A_d >0$.
The weight distribution $A(\C^{\perp})$ of the dual of $\C$ is related to $A(\C)$ by the following identity.

\begin{thm}(The MacWilliams identity) If $\C$ is a linear code over $\F_q$, then
\[
A_{\C^{\perp}} (x, y) = \frac{1}{|\C|} A_{\C} (x+ (q-1)y, x-y).
\]
\end{thm}

In this paper, we are mainly interested in the nonzero weight distribution defined as follows.

\begin{defi}{\em
Let $A(\C)$ be the weight distribution of $\C$. The {\em nonzero weight distribution} of $\C$, denoted by $a(\C)$, is the subsequence of $A(\C)$ consisting of the nonzero values of $A(\C)$ in the order of appearance in $A(\C).$ For instance, the first two elements of $a(\C)$ are $a_0=1$ and $a_1=A_d$.
}
\end{defi}

\begin{defi}{\em
A sequence $a_0, a_1, \dots, a_n$ of real numbers is called {\em log-concave} if $a_i^2 \ge a_{i-1}a_{i+1}$ for all $1 \le i \le n-1$.}
\end{defi}

\begin{defi}{\em
A sequence $A=\{a_k \}_{k=0}^n$ is called {\em unimodal} if there is an index $0 \le j \le n$ such that $a_0 \le \cdots  \le a_{j-1} \le a_j \ge  a_{j+1} \ge \cdots  \ge a_n$.
}
\end{defi}

\begin{defi}{\em
A linear code $\mathcal C$ is called {\em log-concave} if its nonzero weight distribution is log-concave and {\em unimodal} if its nonzero weight distribution is unimodal.}
\end{defi}

The log-concavity property goes back to Newton as the next result shows.

\begin{thm}(\cite{Sta}) \label{thm-binom}
Let $P(x) = \sum_{i=0}^{n} \binom{n}{i}a_i x^i$ be a (real) polynomal with real zeroes. Then $a_j^2 \geqslant a_{j-1} a_{j+1}$ for $j=1, \dots, n-1$.
\end{thm}

\begin{ex}{\em
Let $\C$ be the binary $[4, 3, 2]$ code consisting of all even weight vectors.
Then its weight distribution is 1, 0, 6, 0, 1. Since $0^2 - 1 \cdot 6 <0$, this sequence is not log-concave. However, if we consider only nonzero values 1, 6, 1, then it is log-concave and unimodal. Hence the code is log-concave and unimodal. Also note that $P(x)=1 + \binom{2}{1} \cdot 6 + x^2 = 1 + 12 x + x^2$ has two real roots since $D=12^2 - 4 >0$. This also confirms by Theorem~\ref{thm-binom} that $1, 6, 1$ is log-concave.
}
\end{ex}

In order to consider non log-concave sequences, we introduce the notion of
an $m$-gap log-concave sequence as follows.

\begin{defi}{\rm
	A sequence $A_i$ of integers is {\em $m$-gap log-concave} if $A_i^2 \ge A_{i-1}A_{i+1}$ does not holds for exactly $m$ indices. A linear code $\mathcal{C}$ is {\em $m$-gap log-concave} if its nonzero weight distribution is $m$-gap log-concave.
}
\end{defi}
We note that $0$-gap log-concavity is the same as log-concavity, and that $1$-gap log-concavity means that the log-concavity holds except for one index.

Little is known about the property of the nonzero weight distribution of a linear code.
In particular, we want to know which linear codes are log-concave. As a warm-up, we consider some trivial codes, and show that they are log-concave.

\begin{prop}\label{prop-easy}
\noindent
\begin{enumerate}

\item[{(i)}] The whole space $\mathbb F_2^n$ with $n \geqslant 2$ is log-concave.

\item[{(ii)}] The binary even $[n, n-1]$ code $\mathcal E_n$ for any $n \geqslant 4$ is log-concave.

\end{enumerate}
\end{prop}

\begin{proof}

(i) The weight distribution of $\mathbb F_2^n$ is the binomial distribution
${n \choose 0},  {n \choose 1}, \dots, {n \choose i }, \dots, {n \choose n}$ where all values are positive. Since the binomial distribution is log-concave, $\mathbb F_2^n$ is log-concave.

(ii) Note that $\mathcal E_n$ consists of all even weight vectors of length $n$. So the nonzero weight distribution of $\mathcal E_n$ is the sequence  ${n \choose 0},  {n \choose 2}, \dots, {n \choose 2i }, \dots$. If $n$ is even, then as the all-ones vector is in $\mathcal E_n$,  this sequence is symmetric. Hence we may assume that $2 \leqslant 2i \leqslant n/2 -1$. On the other hand, if $n$ is odd,  we assume that $2 \leqslant 2i \leqslant n-2$.
Then for each case of $n$ we have the following computation.
\[
 {n \choose 2i}^2 \geqslant  {n \choose 2i+2}{n \choose 2i - 2}
    \Leftrightarrow
  \frac{(2i+2)(2i+1)}{(2i)(2i-1)} \cdot \frac{(n-2i+2)(n-2i+1)}{(n-2i)(n-2i-1)} >1.
\]
Hence, the claim of (ii) of Proposition~\ref{prop-easy} is true.
\end{proof}


Next, we investigate the simplex codes, the binary Hamming codes of length 7 and 15, the extended Hamming codes of these lengths, the binary Golay codes of length 23 and 24 because they all belong to famous families of linear codes with known weight distributions.




\begin{ex}{\em
The $q$-ary simplex $[n, r]$ code $\mathcal S_{n,q}$ over $\F_q$ has the nonzero weight distribution 1, $q^r -2$, 1. Hence  $\mathcal S_{n,q}$ is log-concave.
}
\end{ex}

\begin{ex} \label{ex-hamming}{\em
The binary Hamming $[7,4,3]$ code $\mathcal H_3$ has the nonzero weight distribution 1, 7, 7, 1 which is log-concave. Its extended code has the nonzero weight distribution 1, 14, 1 which is log-concave.

However, the binary Hamming $[15, 11, 3]$  code $\mathcal H_4$ has the weight distribution
$$1, 0, 0,  35, 105, 168, 280, 435, 435, 280, 168, 105, 35, 0, 0, 1.$$
Then $168^2 -105\times 280=-1176$, which implies that the nonzero weight distribution of $\mathcal H_4$ is not log-concave. Hence $\mathcal H_4$ is 2-gap log-concave.
On the other hand, its extended code $\bar{\mathcal H_4}$ has the log-concave property because its nonzero weight distribution is $1, 140, 448, 870, 448, 140, 1$.

 The weight distributions of $\mathcal{H}_5$ and $\mathcal{H}_6$ are known as sequences denoted by A010086 and A010087 in OEIS~\cite{Oeis}, respectively. The half of the weight distribution of $\mathcal H_5$ is given in Table~\ref{tab-H5} and the logarithmic distribution is in Figure~\ref{fig-H5}.

\begin{table}
\begin{center}
\begin{tabular}{|c|l|c|l|c|l|}
 \hline Weight $i$ & Frequency $A_i$ &Weight $i$ & Frequency $A_i$ & Weight $i$ & Frequency $A_i$ \\
 \hline
0 & 1 & 7 & 82615  &  12 & 4414865 \\
3 &155 & 8 & 247845 & 13 & 6440560 \\
4 &1085  & 9 & 628680 & 14 & 8280720 \\
5 & 5208 & 10 & 1383096  &  15 & 9398115 \\
6 & 22568 & 11 & 2648919 &   16 & 9398115  \\
\hline
\end{tabular}
\caption{Weight distribution of the binary Hamming code of length 31}
\label{tab-H5}
\end{center}
\end{table}

	\begin{figure}
	\centering
	\includegraphics[scale=0.8]{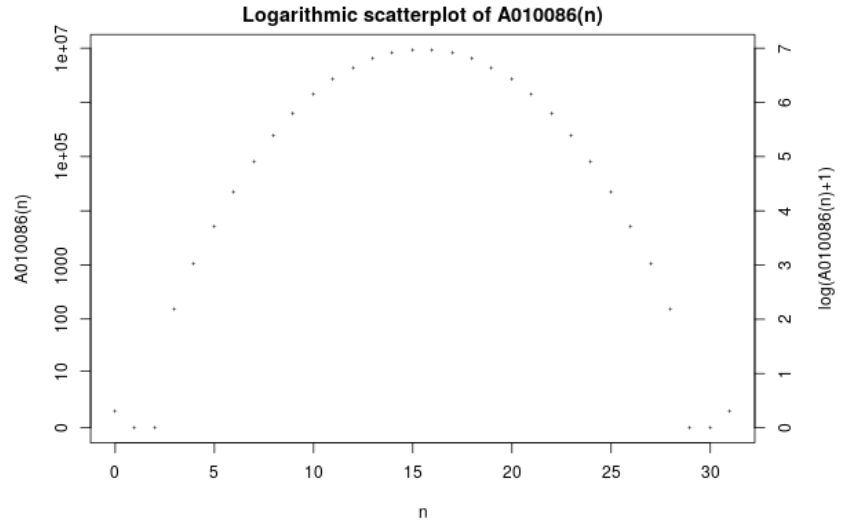}
	\caption{Logarithmic distribution of Table~\ref{tab-H5} from OEIS}
    \label{fig-H5}
	\end{figure}

  One can check easily that each binary Hamming code $\mathcal{H}_m$ for $m=5, 6$ is log- concave.

}
\end{ex}

\begin{ex}{\em
The binary Golay code of length 23 has the weight distribution

$$1, 0, 0, 0, 0, 0, 0, 253, 506, 0, 0, 1288, 1288, 0, 0, 506, 253, 0, 0, 0, 0, 0, 0, 1.$$

Since $506^2 - 253\times 1288=-69828$, the binary Golay code is not log-concave but 2-gap log-concave.

However, the extended binary Golay code of length 24 satisfies the log-concave property since its nonzero weight distribution is
$1, 759, 2576, 759, 1$.
}
\end{ex}

Based on these examples, it appears that there are not many families of linear codes with log-concave nonzero weight distribution. However, in the following sections, we show that the binary Hamming codes (except for the above example), the binary extended Hamming codes, and the second order Reed-Muller codes are all log-concave.




\section{Hamming codes}

	It is well known that the binary Hamming code $\mathcal{H}_{m}$ is the dual of the binary Simplex code $\mathcal{S}_{m}$ and the binary extended Hamming code $\overline{\mathcal{H}}_m$ is the dual of the first order Reed-Muller code.
Using the MacWillams identity, the weight enumerators of $\mathcal{H}_{m}$ and
$\overline{\mathcal{H}}_m$ are known in terms of Krawtchouk polynomials or recursively (see \cite[Example 7.3.4]{HuffmanPless},  \cite[p. 129]{MacSlo}, \cite[Remark 15.1.69]{Handbook-Finite-Field}). For our purpose, we give an explicit formula for these weight enumerators.
\begin{lem}
	Let $\mathcal{H}_m$ and $\overline{\mathcal{H}}_m$ be the binary Hamming $[n=2^m-1, n-m, 3]$ code and its extended code, respectively.
	Then
	\[ (i)~ A_{w}(\mathcal{H}_m) =
	\begin{cases}
		 \frac{1}{n+1} \left( \binom{n}{2i} + (-1)^i n \binom{(n-1)/2}{i} \right) & {\mbox{if }} w = 2i \geqslant 0, \\
		 \frac{1}{n+1} \left( \binom{n}{2i+1} - (-1)^i n \binom{(n-1)/2}{i} \right) & {\mbox{if }} w = 2i+1 \geqslant 1,
	\end{cases}
	\]
	and
	\[ (ii)~ A_{w}(\overline{\mathcal{H}}_m) =
	\begin{cases}
\frac{1}{n} \left( \binom{n}{2i} + (-1)^i (n-1) \binom{n/2}{i} \right) & {\mbox{if }} w = 2i \geqslant 0,  \\
   0 & {\mbox{if }} w = 2i+1 \geqslant 1.
   \end{cases}
   \]
\end{lem}
\begin{proof}
\begin{itemize}
		\item[{(i)}] Since the weight enumerator of $\mathcal S_m$ is $x^n + nx^{(n-1)/2} y^{(n+1)/2}$, the weight enumerator of $\mathcal{H}_m$ is, by the MacWilliams identity,
		\[ W_{\mathcal{H}_m}(x,y) = \frac{1}{n+1} \left( (x+y)^n + n(x+y)^{(n-1)/2} (x-y)^{(n+1)/2} \right), \]
		where $n = 2^m - 1$. By the binomial expansion, we have
		\begin{align*}
			(x+y)^{(n-1)/2} (x-y)^{(n+1)/2} & = (x-y)(x^2 - y^2)^{(n-1)/2} \\
& = (x-y) \sum_{i=0}^{(n-1)/2} \binom{(n-1)/2}{i} (-1)^i x^{n-1-2i} y^{2i} \\
			& = \sum_{i=0}^{(n-1)/2} \binom{(n-1)/2}{i} (-1)^i x^{n-2i} y^{2i} \\
			& ~~~ - \sum_{i=0}^{(n-1)/2} \binom{(n-1)/2}{i} (-1)^i x^{n-1-2i} y^{2i+1}.
		\end{align*}
		Hence, \[ A_{2i}(\mathcal{H}_m) = \frac{1}{n+1} \left( \binom{n}{2i} + (-1)^i n \binom{(n-1)/2}{i} \right) \]
		and \[ A_{2i+1}(\mathcal{H}_m) = \frac{1}{n+1} \left( \binom{n}{2i+1} - (-1)^i n \binom{(n-1)/2}{i} \right). \]
		
		\item[{(ii)}] Since the weight enumerator of the first order Reed-Muller code $\mathcal{R}(1,m)$ is
		\[ x^n + y^n + (2n-2) x^{n/2} y^{n/2}, \]
		where $n = 2^m$, the weight enumerator of $\overline{\mathcal{H}}_m$ is
		\[ W_{\overline{\mathcal{H}}_m}(x,y) = \frac{1}{2n} \left( (x+y)^n + (x-y)^n + (2n-2) (x+y)^{n/2} (x-y)^{n/2} \right). \]
		That is to say $A_{2i+1}(\overline{\mathcal{H}}_m) = 0$ and
		\[ A_{2i}(\overline{\mathcal{H}}_m) = \frac{1}{n} \left( \binom{n}{2i} + (-1)^i (n-1) \binom{n/2}{i} \right). \]
	\end{itemize}
	This completes the proof.
\end{proof}

The following lemma lists some well known properties of binomial coefficients which are used in the proofs of the next theorems.

\begin{lem} \label{lem-nchoosek}
	 We have the following properties for the binomial coefficients:
	\begin{enumerate}
		\item[(i)] If $1 \leqslant k \leqslant n-1$, then we have the Pascal's Triangle
		\[ \binom{n}{k} = \binom{n-1}{k} + \binom{n-1}{k-1}. \]
		
		\item[(ii)] If $1 \leqslant k \leqslant n-1$, then
		\[ \binom{n}{k} = \frac{n-k+1}{k} \binom{n}{k-1} = \frac{n}{n-k} \binom{n-1}{k} \quad \textnormal{and} \quad k\binom{n}{k} = n\binom{n-1}{k-1}. \]
		
		\item[(iii)] For given positive integers $m$ and $n$, we have the Vandermonde identity
		\[ \sum_{r=0}^{k} \binom{m}{r} \binom{n}{k-r} = \binom{m+n}{k}. \]
		In particular, for any odd $n$ and positive integer $i$,
		\[ \binom{n}{2i} > \binom{(n-1)/2}{i} \binom{(n+1)/2}{i}. \]
	\end{enumerate}
\end{lem}

In Example~\ref{ex-hamming} we showed that if $m=3, 5, 6$ then $\mathcal{H}_m$ is log-concave. Furthermore, we show that $\mathcal{H}_m$ is log-concave for any $m \geqslant 7$.
\begin{thm}
	If $m=3$ or $m \geqslant 5$, the binary Hamming code $\mathcal{H}_m$ is log-concave.
\end{thm}
\begin{proof}
Let  $n=2^m-1$.	We may assume that $m \geqslant 7$. Letting $B_i = (n+1)A_i(\mathcal{H}_m)$ for all $i \leqslant n$, we have $B_i = 0$ if and only if $i = 1,2,n-1,n-2$ (\cite[Ch.6, p.157, E2]{MacSlo}).
 Since the all-ones vector is in $\mathcal{H}_m$, the sequence $B_i$ is symmetric, hence, $B_{2i} = B_{n-2i}$.
 If $2 \leqslant i \leqslant \lfloor n/2 \rfloor$, then since $n$ is odd,
it is enough to show that $B_{2i}^2 - B_{2i-1} B_{2i+1} >0$ in order to show the sequence of $A_i$s is log-concave.

First, we claim that if $m \geqslant 4$, then $A_3^2(\mathcal{H}_m) \geqslant A_4(\mathcal{H}_m) A_0(\mathcal{H}_m)$, where
\[ A_3(\mathcal{H}_m) = \frac{\binom{n}{3} + \binom{n}{2}}{n+1}, \ A_4(\mathcal{H}_m) = \frac{\binom{n}{4}+n\binom{(n-1)/2}{2}}{n+1}. \]
That is, we need to show
\[ \left( \binom{n}{3} + \binom{n}{2} \right)^2 \geqslant (n+1)\left( \binom{n}{4} + n \binom{(n-1)/2}{2}\right). \]
Let $n'' = (n-1)/2$. If $n \geqslant 5$, we have $(n-2)(n-3) \geqslant 3(n-3) = 6(n''-1)$ and $n-1= 2n''$.
Thus,
\[ \binom{n}{4} = \frac{n(n-1)(n-2)(n-3)}{24} \geqslant \frac{nn''(n''-1)}{2} = n \binom{(n-1)/2}{2}. \]
Also, if $n \geqslant 5$, we have $n(n-1) \geqslant 3(n+1)$ and
\[ \binom{n}{3} = \frac{n(n-1)(n-2)}{6} \geqslant \frac{3(n+1)(n-2)}{6} > \frac{(n+1)(n-3)}{2}. \]
Combining these two inequalities and using Lemma \ref{lem-nchoosek} (ii), we have
\begin{align*}
	\left( \binom{n}{3} + \binom{n}{2} \right)^2 & > \binom{n}{3}^2 = \frac{4}{n-3} \binom{n}{3} \binom{n}{4} > 2(n+1)\binom{n}{4} \\
	& > (n+1)\left( \binom{n}{4} + n \binom{(n-1)/2}{2}\right),
\end{align*}
which implies that $A_3^2(\mathcal{H}_m) \geqslant A_4(\mathcal{H}_m) A_0(\mathcal{H}_m)$ if $m \geqslant 4$.

Next, as for $2 \leqslant i \leqslant \lfloor n/2 \rfloor -2$, i.e., $4 \leqslant w = 2i \leqslant n-4$, we have
	\begin{align*}
		g(i) & = B_{2i}^2 - B_{2i-1} B_{2i+1} = \binom{n}{2i}^2 + (-1)^i 2n \binom{(n-1)/2}{i} \binom{n}{2i} + n^2  \binom{(n-1)/2}{i}^2 \\
		& - \left( \binom{n}{2i-1} + (-1)^i n \binom{(n-1)/2}{i-1} \right) \left( \binom{n}{2i+1} - (-1)^i n \binom{(n-1)/2}{i} \right) \\
		& = \binom{n}{2i}^2 - \binom{n}{2i-1} \binom{n}{2i+1} + n^2 \binom{(n-1)/2}{i}^2 + n^2 \binom{(n-1)/2}{i-1} \binom{(n-1)/2}{i} \\
		& +(-1)^i n \left( 2\binom{(n-1)/2}{i} \binom{n}{2i} - \binom{(n-1)/2}{i-1} \binom{n}{2i+1} + \binom{(n-1)/2}{i} \binom{n}{2i-1} \right).
	\end{align*}
	Since $\binom{n}{2i}^2 - \binom{n}{2i-1} \binom{n}{2i+1} >0$, we have by Lemma \ref{lem-nchoosek} (ii)
	\begin{align*}
		&2\binom{(n-1)/2}{i} \binom{n}{2i} - \binom{(n-1)/2}{i-1} \binom{n}{2i+1} \\ & = \left( \frac{2((n-1)/2 - i +1)}{i} - \frac{n-2i}{2i+1} \right) \binom{(n-1)/2}{i-1} \binom{n}{2i} \\
		& = \left( \frac{n - 2i + 1}{i} - \frac{n-2i}{2i+1} \right) \binom{(n-1)/2}{i-1} \binom{n}{2i} > 0.
	\end{align*}
	Hence, if $i$ is even, $g(i) > 0$.

Finally let us suppose that $i$ is an odd integer $\geqslant 3$ because if $i=1$ then we have $B_2 =(n+1)A_2=0$.
Let
\[ h(i) = \binom{n}{2i}^2 - \binom{n}{2i-1} \binom{n}{2i+1} - n \left( 2\binom{(n-1)/2}{i} \binom{n}{2i} + \binom{(n-1)/2}{i} \binom{n}{2i-1}\right). \]
Compared with each term of $g(i)$, it is clear that $g(i) > 0$ if $h(i) > 0$.
In order to estimate $h(i) > 0$, let $n' = (n+1)/2$, and we need the following inequality
\begin{equation} \label{eq-main1}
	\frac{1}{n-4} \binom{(n+1)/2}{3} - 2n > 0
\end{equation}
which holds if and only if
\begin{equation} \label{eq-main2}
	\binom{n'}{3} = \frac{n'(n'-1)(n'-2)}{6} > 2n(n-4) = 2(2n'-1)(2n'-5).
\end{equation}
Suppose that $n' > 48$ (that is to say, $m \geqslant 7$). Then it is easy to see that
\[ (n'-1)(n'-2) \geqslant 48(n'-2) = 24(2n'-5) + 24 > 24(2n'-5), \]
which implies
\[
\frac{n'(n'-1)(n'-2)}{6} > \frac{n'}{6} \cdot 24 \cdot (2n'-5) = 4n' (2n' -5) >  2(2n'-1)(2n'-5).
\]
Therefore, Equation~(\ref{eq-main2}) is true, hence Equation~(\ref{eq-main1}) is true as wanted.
Since $i \geqslant 3$ is odd, $A_{n-1} = A_{n-2} = 0$, and $n - 3 = 2^m - 4$ is doubly even, hence $2i \leqslant n - 5$.
Thus,
	\begin{align*}
		h(i) & = \binom{n}{2i}^2 - \binom{n}{2i-1} \binom{n}{2i+1} - n \left( 2\binom{(n-1)/2}{i} \binom{n}{2i} + \binom{(n-1)/2}{i} \binom{n}{2i-1}\right) \\
		& = \left( 1 - \frac{2i(n-2i)}{(2i+1)(n-2i+1)} \right) \binom{n}{2i}^2 - n \binom{(n-1)/2}{i} \left( \binom{n+1}{2i} + \binom{n}{2i} \right)  \\
		& = \frac{n+1}{(2i+1)(n-2i+1)} \binom{n}{2i}^2 - n \binom{(n-1)/2}{i} \binom{n}{2i} \left( 1 + \frac{n+1}{n+1-2i} \right) \\
		& > \frac{n+1}{(2i+1)(n-2i+1)} \binom{n}{2i} - n \binom{(n-1)/2}{i} \left( 1 + \frac{n+1}{n+1-2i} \right) \\
		& = \frac{n+1}{n+1-2i} \left( \frac{1}{2i+1} \binom{n}{2i} - n \binom{(n-1)/2}{i} \right) - n \binom{(n-1)/2}{i} \\
		& > \frac{1}{2i+1} \binom{n}{2i} - n \binom{(n-1)/2}{i} - n \binom{(n-1)/2}{i} \\
		& > \binom{(n-1)/2}{i} \left( \frac{1}{2i+1} \binom{(n+1)/2}{i} - 2n \right)
		> \frac{1}{2i+1} \binom{(n+1)/2}{i} - 2n \\
		& \geqslant \frac{1}{n-4} \binom{(n+1)/2}{3} - 2n > 0 {\mbox{~by Equation}}~(\ref{eq-main1}) ,
	\end{align*}
where we used Lemma \ref{lem-nchoosek} (i) in the second line of $h(i)$, Lemma \ref{lem-nchoosek} (ii) in the third line of $h(i)$, and Lemma \ref{lem-nchoosek} (iii) in the 7-th line of $h(i)$.

As a summary, if $m \geqslant 7$, $g(i) > 0$ and $B_{2i}^2 - B_{2i-1} B_{2i+1} > 0$ for all possible $i$.
Since the weight distribution of $\mathcal{H}_m$ is symmetric, $A_w^2 - A_{w-1}A_{w+1} > 0$ for all $4 \leqslant w \leqslant n - 4$.
	
Therefore, if $m \geqslant 7$, $\mathcal{H}_m$ is log-concave.
\end{proof}

In Example~\ref{ex-hamming}, we mentioned that the binary extended Hamming code $\mathcal{H}_m$ for $m=3$ or $4$ is log-concave. Interestingly, it is also true for any $m \geqslant 5$.

\begin{thm}
	The binary extended Hamming code $\overline{\mathcal{H}}_m$ ($m \geqslant 3$) is log- concave.
\end{thm}
\begin{proof}
	It is well known~\cite{MacSlo} that for $m \geqslant 1$, the weight distribution of the first order Reed-Muller code $\mathcal{R}(1,m)$ is \[ A_0 = 1, A_{n/2} = 2n-2, A_n = 1, \ \textnormal{where} \ n = 2^m. \]
	Then by the MacWilliams identity, the weight distribution $A'_i$ of $\overline{\mathcal{H}}_m$ satisfies
	\[ n A'_{2i} = \binom{n}{2i} + (-1)^i (n-1) \binom{n/2}{i} = B_i. \]
	It is clear that $B_i = B_{n/2-i}$.
	Moreover, if $m \geqslant 4$ and $2 \leqslant i \leqslant n/4$,
	\begin{align*}
		B_i & \geqslant  \binom{n}{2i} - (n-1) \binom{n/2}{i} > \binom{n/2}{i} \binom{n/2}{i} -(n-1) \binom{n/2}{i} \\
		& \geqslant \binom{n/2}{i} \left(\binom{n/2}{2}-(n-1)\right) > \binom{n/2}{i} > 0.
	\end{align*}
	
Hence $B_i$'s are all positive for $0 \leqslant i \leqslant n/2$.

	
In what follows, we consider $i$ to be even or odd since the computation depends on the parity of $i$.

	Suppose that $3 \leqslant i \leqslant n/4-1$ is even. Then we have
	\[ \frac{B_i^2}{B_{i-1} B_{i+1}} \geqslant \frac{\binom{n}{2i}^2}{\binom{n}{2i-2} \binom{n}{2i+2}} \geqslant 1. \]
	We also have
	\[ B_2^2 \geqslant B_0 B_3. \]

Now suppose that $i$ is odd. Then it is easy to check that
	\[ B_i^2 = \binom{n}{2i} \binom{n}{2i} -2(n-1)\binom{n/2}{i} \binom{n}{2i} + (n-1)^2 \binom{n/2}{i} \binom{n/2}{i} \]
	can be regarded as a polynomial on $n$, with the degree $4i$, and the leading coefficient $\left(1/((2i)!) \right)^2$.
	Similarly,
	\begin{align*}
	B_{i-1} B_{i+1} &= \binom{n}{2i-2} \binom{n}{2i+2} + (n-1)\left( \binom{n/2}{i-1} \binom{n}{2i+2} + \binom{n/2}{i+1} \binom{n}{2i-2} \right) \\& + (n-1)^2 \binom{n/2}{i-1} \binom{n/2}{i+1}
	\end{align*}
	can be regarded as a polynomial on $n$ with the degree $4i$, and the leading coefficient \[ \frac{1}{(2i-2)! (2i+2)!}. \]
	Since
	\[ \frac{(2i)! (2i)!}{(2i-2)! (2i+2)!} = \frac{(2i)(2i-1)}{(2i+1)(2i+2)} < 1, \]
	we have $B_{i-1} B_{i+1} < B_i^2$, if $n$ is large enough.
	That is to say, $\overline{\mathcal{H}}_m$ is also log-concave.
	
	In fact, $B_i^2 - B_{i-1}B_{i+1} > 0$ if
	\begin{align*}
		& \binom{n}{2i}^2 - \binom{n}{2i-2} \binom{n}{2i+2} \\
		&-  (n-1) \left(2 \binom{n/2}{i} \binom{n}{2i}
		+ \binom{n/2}{i-1} \binom{n}{2i+2} + \binom{n/2}{i+1} \binom{n}{2i-2} \right) \\
		> &\binom{n}{2i}^2 - \binom{n}{2i-2} \binom{n}{2i+2} - 4(n-1) \binom{n/2}{i+1} \binom{n}{2i+2} \\
		= &\binom{n}{2i+2} \left( \frac{(2i+1)(2i+2)}{(n-2i)(n-2i-1)} \binom{n}{2i} - \binom{n}{2i-2} - 4(n-1) \binom{n/2}{i+1} \right) > 0.
	\end{align*}
	Let $n' = n/2$. We have
	\begin{align*}
		& \frac{(2i+1)(2i+2)}{(n-2i)(n-2i-1)} \binom{n}{2i} - \binom{n}{2i-2} - 4(n-1) \binom{n/2}{i+1} \\
		= & \frac{(2i+1)(2i+2)(n-2i+1)(n-2i+2)}{2i(2i-1)(n-2i)(n-2i-1)} \binom{n}{2i-2} - \binom{n}{2i-2} - 4(n-1) \binom{n/2}{i+1} \\
		> & \left( \frac{(2i+1)(2i+2)}{2i(2i-1)} - 1\right) \binom{n}{2i-2} - 4(n-1) \binom{n/2}{i+1} \\
		> & \frac{8i+2}{4i^2-2i} \binom{n}{2i-2} - 4(n-1) \binom{n/2}{i+1} > \frac{8i+2}{4i^2-2i} \binom{n/2}{i-1}^2 - 4(n-1) \binom{n/2}{i+1} \\
		> & \binom{n/2}{i-1} \left( \frac{8i+2}{4i^2-2i} \binom{n/2}{i-1} - 4(n-1) \frac{(n'-i+1)(n'-i)}{i(i+1)} \right).
	\end{align*}
	It is easy to see that if $i \geqslant 5$ and $n' > 16$,
	\[ \frac{3 \binom{n'}{i-1}}{(n-1)(n'-i+1)(n'-i)} \geqslant \frac{3 \binom{n'}{4}}{(n-1)(n'-4)(n'-5)} > \frac{n'(n'-1)}{8(2n'-1) } > 1, \]
	which is equivalent to $n'^2-17n'+8 > 0$.
	Besides, if $n' = 16$, then $3\binom{16}{4}-31\times 11 \times 12 = 1368$.
	Hence, if $i \geqslant 5$ and $n \geqslant 32$, we have
	\begin{align*}
		 & \frac{i(i+1)(8i+2)}{4i^2-2i} \binom{n/2}{i-1} - 4(n-1)(n'-i+1)(n'-1) \\
		 & > 2(i+1) \binom{n/2}{i-1} - 4(n-1)(n'-i+1)(n'-1) \\
		 & \geqslant 4 \left( 3 \binom{n'}{i-1} > (n-1)(n'-i+1)(n'-i) \right) > 0.
	\end{align*}
	In a word, if $i \geqslant 5$ and $m \geqslant 5$, then $B_i^2 - B_{i-1}B_{i+1} > 0$.
	
	As for $i = 3$, considering that
	\[ \frac{\binom{n/2}{3} \binom{n}{6}}{\binom{n/2}{4} \binom{n}{4}} = \frac{4(n-5)(n-6)}{30(n'-4)}, \ \textnormal{and} \ \frac{\binom{n'}{2} \binom{n}{8}}{\binom{n/2}{3} \binom{n}{6}} = \frac{3(n-7)(n-8)}{56(n-3)}, \]
	we have that if $n \geqslant 32$,
	\[ \binom{n/2}{4} \binom{n}{4} < \binom{n/2}{3} \binom{n}{6} < \binom{n'}{2} \binom{n}{8}. \]
	Similar to the discussions above, we have
	\begin{align*}
		B_3^2-B_2B_4 & > \binom{n}{6}^2 - \binom{n}{4} \binom{n}{8} - 4(n-1) \binom{n'}{2} \binom{n}{8} \\
		& = \binom{n}{8} \left( \frac{56}{(n-6)(n-7)} \binom{n}{6} - \binom{n}{4} - \frac{n(n-1)(n-4)}{2}\right) \\
		& = \binom{n}{8} \left( \frac{56(n-5)(n-4)}{30(n-7)(n-6)} \binom{n}{4} - \binom{n}{4} - \frac{n(n-1)(n-4)}{2}\right) \\
		& > \binom{n}{8} \left( \frac{4}{5} \binom{n}{4} - \frac{n(n-1)(n-4)}{2}\right) > 0,
	\end{align*}
	since
	\[ \frac{4}{5} \binom{n}{4} - \frac{n(n-1)(n-4)}{2} >0  \Leftrightarrow \frac{8\binom{n}{4}}{5n(n-1)(n-4)} = \frac{(n-2)(n-3)}{15(n-4)} > 1~~
\textnormal{as} \ n \geqslant 32. \]
	Therefore, if $m \geqslant 5$, then $\overline{\mathcal{H}}_m$ is log-concave.
\end{proof}

\begin{prop}
	Let $\mathcal{H}_{m,q}$ be the $q$-ary Hamming code. Then
	\[ q^m A_w(\mathcal{H}_{m,q}) = (q-1)^w \binom{n}{w} + n(q-1) \sum_{i+j=w} \binom{(n-1)/q}{i} \binom{((q-1)n+1)/q}{j} (q-1)^i (-1)^j.  \]
\end{prop}
\begin{proof}
	The weight enumerator of $\mathcal{S}_{m,q}$ is
	\[ W_{\mathcal{S}_{m,q}}(x,y) = x^n + (q^m-1) x^{n-q^{m-1}} y^{q^{m-1}}, \]
	where $n = (q^m-1)/(q-1)$.
	Then the weight enumerator of $\mathcal{H}_{m,q}$ is
	\[ W_{\mathcal{H}_{m,q}}(x,y) = \frac{1}{q^m} \left( (x+(q-1)y)^n + (q^m-1)(x+(q-1)y)^{n-q^{m-1}} (x-y)^{q^{m-1}} \right). \]
	Hence,
	\begin{align*}
		q^m A_w(\mathcal{H}_{m,q}) & = (q-1)^w \binom{n}{w} + (q^m-1) \sum_{i+j=w} \binom{n-q^{m-1}}{i} \binom{q^{m-1}}{j} (q-1)^i (-1)^j \\
		& = (q-1)^w \binom{n}{w} + n(q-1) \sum_{i+j=w} \binom{(n-1)/q}{i} \binom{((q-1)n+1)/q}{j} (q-1)^i (-1)^j.
	\end{align*}
	This completes the proof.
\end{proof}

\begin{rem}
	It is known that $A_w(\mathcal{H}_{m,q}) > 0$ for all $w \geqslant 3$, except for the case $q = 3$ and $m = 2$ (see \cite[Theorem 38]{Hamming-weight-number}).
\end{rem}

We have checked by Magma that if $q \leqslant 8$, the $q$-ary Hamming code $\mathcal{H}_{2,q}$ is not log-concave.

\begin{thm} \label{thm-MDS-Hamming}
	If $q \geqslant 9$,  the $q$-ary Hamming code $\mathcal{H}_{2,q}$ is log-concave.
\end{thm}

\begin{proof}
	It is known that $\mathcal{H}_{2,q}$ is an $[n = q+1, k = q-1, 3]_q$ MDS code.
	We have checked by Magma that $\mathcal{H}_{2,9}$ and $\mathcal{H}_{2,11}$ are log- concave. So we assume that $q \geqslant 13$.
	For each $4 \leqslant w \leqslant n-1$, according to Lemma \ref{lemma: MDS-g}, $g_k(s,n,q) > 1$ if $s = w-3$ is even.
	That is to say, if $w$ is odd, then $A_w^2 > A_{w-1} A_{w+1}$.
	
	As for odd $s$, consider
	\[ g_k(1,4,q) = \frac{(q-3)^2}{q^2-4q+6} \geqslant \frac{s+3}{s+4}, \]
	which is equivalent to
	\[ q^2 - (2s+12) q + 3s+18 \geqslant 0. \]
	Noting that its discriminant is
	\[ \Delta = (2s+12)^2 - 4(3s+12) = 4(s+6)^2 - 12(s+6) = 4(s+3)(s+6) \]
	and
	\[ 4(s+4)^2 = 4(s^2 + 8s + 16) < \Delta = 4(s^2 + 9s + 18) < 4(s^2 + 10s + 25) = 4(s+5)^2, \]
	then we have
	\[ s+4 < \frac{\sqrt{\Delta}}{2} = \sqrt{(s+6)(s+3)} < s+5. \]
	Since $q \geqslant 9$ is an integer, then $q \geqslant s+6 + \sqrt{(s+6)(s+3)}$, i.e., $q \geqslant 2s+11$.
	
	If $1 \leqslant s \leqslant (q-11)/2$, then
	\[ \frac{A_{w}^2}{A_{w-1}A_{w+1}} = \frac{4+s}{3+s} \cdot \frac{k-s}{k-1-s} \cdot g_k(s,w,q) > \frac{4+s}{3+s} \cdot g_k(1,4,q) > 1. \]
	According to the proof of Theorem \ref{thm: MDS-log}, if $s \geqslant (q-7)/2$, we have
	\begin{equation} \label{eq-MDS-chain}
		\frac{A^2_{w}}{A_{w-1} A_{w+1}} < \frac{A^2_{w+2}}{A_{w+1} A_{w+3}} < \frac{A^2_{w+4}}{A_{w+3} A_{w+5}} < \dots.
	\end{equation}
	
	If $q \equiv 1 \pmod{4}$ is odd, then $(q-9)/2$ is even and $(q-7)/2$ is odd.
	Thus we only need to check the case when $s = (q-7)/2$ and $w = s+3$,
	\begin{align*}
		\frac{A_{w}^2}{A_{w-1}A_{w+1}} & = \frac{4+s}{3+s} \cdot \frac{k-s}{k-1-s} \cdot g_k(s,w,q)
		> \frac{q+1}{q-1} \cdot \frac{q+5}{q+3} \cdot \frac{(q-3)^2}{q^2-4q+6} \\
		& = \frac{q^4 - 22 q^2 + 24 q +45}{q^4 - 2q^3 - 5q^2 + 24q - 18} > 1
	\end{align*}
	if $2q^3 + 5q^2 - 24q + 18 > 22q^2 -24q - 45$, i.e., $2q^3 - 17q^2 + 63 > 0$.
	In fact, if $q \geqslant 9$, then the last inequality holds, hence $\mathcal{H}_{2,q}$ is log-concave.
	
	If $q \equiv 3 \pmod{4}$, then $(q-9)/2$ is odd and $(q-7)/2$ is even.
	As for $s = (q-9)/2$ and $w = s+3$, we have
	\begin{align*}
		\frac{A_{w}^2}{A_{w-1}A_{w+1}} & = \frac{4+s}{3+s} \cdot \frac{k-s}{k-1-s} \cdot g_k(s,w,q)
		> \frac{q-1}{q-3} \cdot \frac{q+7}{q+5} \cdot \frac{(q-3)^2}{q^2-4q+6} \\
		& = \frac{q^4 - 34q^2 + 96q - 63}{q^4 - 2q^3 - 17q^2 + 72q - 90} > 1
	\end{align*}
	if $q^4 - 34q^2 + 96q - 63 > q^4 - 2q^3 - 17q^2 + 72q - 90$, i.e., $2q^3 - 17q^2 + 24q + 27 > 0$.
	Also, if $q \geqslant 9$, then the above inequalities hold.
	Similarly, when $s = (q-9)/2 + 2 = (q-5)/2$ and $w = s+3$, we have
	\begin{align*}
		\frac{A_{w}^2}{A_{w-1}A_{w+1}} & = \frac{4+s}{3+s} \cdot \frac{k-s}{k-1-s} \cdot g_k(s,w,q)
		> \frac{q+3}{q+1} \cdot \frac{q+3}{q+1} \cdot \frac{(q-3)^2}{q^2-4q+6} \\
		& = \frac{q^4 - 18q^2 + 81}{q^4 - 2q^3 - q^2 + 8q + 6} > 1,
	\end{align*}
	if $2q^3 - 17q^2 - 8q + 75 > 0$, which holds when $q \geqslant 9$.
	Thus, by the inequality \eqref{eq-MDS-chain}, if $q \equiv 3 \pmod{4}$, then $\mathcal{H}_{2,q}$ is log-concave.

	If $q = 2^m$ is even, then $(q-10)/2$ and $(q-6)/2$ are odd.
	If $s = (q-10)/2$, then
	\begin{align*}
		\frac{A_{w}^2}{A_{w-1}A_{w+1}} & = \frac{4+s}{3+s} \cdot \frac{k-s}{k-1-s} \cdot g_k(s,w,q)
		> \frac{q-2}{q-4} \cdot \frac{q+8}{q+6} \cdot \frac{(q-3)^2}{q^2-4q+6} \\
		& = \frac{q^4 - 43q^2 + 150q - 144}{q^4 - 2q^3 - 26q^2 + 108q - 144} > 1,
	\end{align*}
	if $2q^3 - 17q^2 + 42q > 0$, which holds when $q \geqslant 9$.
	As for $s = (q-6)/2$, we have
	\begin{align*}
		\frac{A_{w}^2}{A_{w-1}A_{w+1}} & = \frac{4+s}{3+s} \cdot \frac{k-s}{k-1-s} \cdot g_k(s,w,q) \cdot
		> \frac{q+2}{q} \cdot \frac{q+4}{q+2} \cdot \frac{(q-3)^2}{q^2-4q+6} \\
		& = \frac{q^4 - 19q^2 + 6q + 72}{q^4 - 2q^3 - 2q^2 + 12q} > 1,
	\end{align*}
	if $2q^3 - 17q^2 - 6q + 72 > 0$, which holds when $q \geqslant 9$.
	Using the inequality~\eqref{eq-MDS-chain} again, for even $q \geqslant 16$, $\mathcal{H}_{2,q}$ is log-concave.

	Therefore, for all prime power $q \geqslant 9$, $\mathcal{H}_{2,q}$ is log-concave.
\end{proof}

\section{The second order Reed-Muller codes}

In this section, we show that the first and second order binary Reed-Muller codes are log-concave. We then show that the homogeneous and projective second order $q$-ary Reed Muller codes are $1$-gap log-concave.

\begin{prop}
The Reed-Muller codes $\mathcal R(r, m)$ are log-concave for $1 \le r \le m \le 6$.
\end{prop}

\begin{proof}
We have checked this by Magma since the weight distribution of $R(r, m)$ for these small values of $r$ and $m$ are known.
\end{proof}

\begin{prop}
The first order Reed-Muller code $\mathcal R(1, m)$ is log-concave for any $m \ge 1$.
\end{prop}

\begin{proof}
The code $\mathcal R(1, m)$ has three nonzero numbers $A_0 =1,  A_{2^{m-1}}= 2^{1+m} -2, A_{2^m} =1$. Hence $R(1, m)$ is log-concave.
\end{proof}

The weight distribution of $\mathcal R(2, m)$ is given by S. Li~\cite{Li}. We manipulated the formula to prove the following

\begin{lem}(\cite{Li}, \cite[Ch.8. Theorem 8]{MacSlo})
	The weight distribution of the second order Reed-Muller code $\mathcal R(2, m)$ is given below.
\begin{center}
\begin{tabular}{|c|c|}
 \hline Weight $w$ & Frquency $A_w$\\ \hline
	$0$ & $1$ \\[2ex] $2^{m-1}$ & $2\left(2^m-1+\sum_{\ell=1}^{\lfloor {m-1\over 2}\rfloor}(2^{\ell^2+\ell}){\prod_{i=m-2\ell}^m(2^i-1)\over \prod_{i=1}^\ell(2^{2i}-1)}\right)$\\[2ex]
	$2^{m-1}\pm 2^{m-j-1}, 1\le j\le \lfloor {m\over 2}\rfloor$ & $(2^{j^2+j}){\prod_{i=m-2j+1}^m(2^i-1)\over \prod_{i=1}^j(2^{2i}-1)}$\\[2ex]
	$2^m$ & $1$\\ \hline
\end{tabular}
\end{center}
\end{lem}

\begin{thm}
The second order Reed-Muller code $\mathcal R(2, m)$ is log-concave for $m \ge 2$.
\end{thm}

\begin{proof}
We first show that the nonzero weight distribution of $\mathcal R(2, m)$ is unimodal with largest term $A_{2^{m-1}}$.
For $1\le j\le \lfloor{m\over 2}\rfloor-1$, we have
\begin{equation*}
	A_{2^{m-1}-2^{m-(j+1)-1}}/A_{2^{m-1}-2^{m-j-1}}={2^{2j+2}\over 2^{2j+2}-1}(2^{m-2j-1}-1)(2^{m-2j}-1)\ge 1.
\end{equation*}
To show that $A_{2^{m-1}}$ is the largest term, we consider the cases when $m$ is even and  odd separately.
If $m$ is odd, then
	\[A_{2^{m-1}- 2^{m-\left\lfloor{m\over 2}\right\rfloor-1}}\le 2^{{m-1\over 2}^2+{m-1\over 2}}{\prod_{i=2}^m(2^i-1)\over \prod_{i=1}^{m-1\over 2}(2^{2i}-1)}\le A_{2^m-1}.\]
On the other hand, if $m$ is even, then
\begin{align*}
	A_{2^{m-1}- 2^{m-\left\lfloor{m\over 2}\right\rfloor-1}}&\le 2^{{m\over 2}^2+{m\over 2}}{\prod_{i=1}^m(2^i-1)\over \prod_{i=1}^{m\over 2}(2^{2i}-1)}\\&=2^{{m\over 2}^2-{m\over 2}}{\prod_{i=1}^{m-1}(2^i-1)\over \prod_{i=1}^{{m\over 2}-1}(2^{2i}-1)}2^m\\&\le 2\left(2^{{m\over 2}^2-{m\over 2}}{\prod_{i=1}^{m}(2^i-1)\over \prod_{i=1}^{{m\over 2}-1}(2^{2i}-1)}\right)\le A_{2^m-1}.
\end{align*}
Since the weight distribution of $\mathcal R(2,m)$ is symmetric with respect to $A_{2^{m-1}}$, it is sufficient to show that the sequence $A_{2^{m-1}- 2^{m-j-1}}$ is log-concave for $1\le j\le \lfloor {m\over 2}\rfloor$. We consider three cases.

\medskip

Case 1) Let $j=1$. Then
\begin{align*}
A_{2^{m-1}-2^{m-3}}&={64\over 21}(2^{m-3}-1)(2^{m-2}-1)(2^{m-1}-1)(2^{m}-1)\\&={16\over 7}(2^{m-3}-1)(2^{m-2}-1)A_{2^{m-1}-A^{m-2}}\\&\le (A_{2^{m-1}-A^{m-2}})^2.
\end{align*}

Case 2) \noindent Suppose $2\le j\le \lfloor{m\over 2}\rfloor-1$. Then we have
\begin{align*}
	&A_{2^{m-1}- 2^{m-(j+1)-1}}A_{2^{m-1}- 2^{m-(j-1)-1}}\\&=(A_{2^{m-1}- 2^{m-j-1}})^2\left({2^{2j+2}-4\over 2^{2j+2}-1}\right){\prod_{i=m-2j-1}^{m-2j}(2^i-1)\over \prod_{i=m-2j+1}^{m-2j+2}(2^{2i}-1)}\\&\le (A_{2^{m-1}- 2^{m-j-1}})^2.
\end{align*}

Case 3) Suppose $j=\lfloor {m\over 2}\rfloor$. We consider two subcases as follows.

Firstly, we assume that $m$ is odd. We want to show that
\[A_{2^{m-1}}A_{2^{m-1}-2^{m-{m-1\over 2}}}\le A_{2^{m-1}-2^{m-{m-1\over 2}-1}}^2,\]
that is,
\[A_{2^{m-1}}\le A_{2^{m-1}-2^{m-{m-1\over 2}-1}}^2/A_{2^{m-1}-2^{m-{m-1\over 2}}}.\]
Since
\begin{align*}
	&A_{2^{m-1}-2^{m-{m-1\over 2}-1}}/A_{2^{m-1}-2^{m-{m-1\over 2}}}\\&=2^{({m-1\over 2})^2+({m-1\over 2})}{\prod_{i=2}^m(2^i-1)\over \prod_{i=1}^{m-1\over 2}(2^{2i}-1)}\bigg/2^{({m-3\over 2})^2+({m-3\over 2})}{\prod_{i=4}^m(2^i-1)\over \prod_{i=1}^{m-3\over 2}(2^{2i}-1)}\\&=2^{2({m-1\over 2})}\cdot {(2^2-1)(2^3-1)\over 2^{2({m-1\over 2})}-1}\\&=21\cdot {2^{m-1}\over 2^{m-1}-1},
\end{align*}
we only need to show that
\begin{equation*}
A_{2^{m-1}}\le 21\cdot {2^{m-1}\over 2^{m-1}-1}\cdot A_{2^{m-1}-2^{m-{m-1\over 2}-1}}.
\end{equation*}
Let $B_\ell=(2^{m-2\ell}-1)\cdot A_{2^{m-1}-2^{m-\ell-1}}$ for $1\le \ell \le {m-1\over 2}$. Then
\[A_{2^{m-1}}=2\left(2^m-1+\sum_{\ell=1}^{{m-1\over 2}}B_\ell\right).\]
For $1\le \ell\le {m-1\over 2}-1$,
\begin{align*}
B_{\ell+1}/B_\ell	&=\left({2^{m-2\ell-2}-1\over 2^{m-2\ell}-1}\right)\cdot A_{2^{m-1}-2^{m-(\ell+1)-1}}/A_{2^{m-1}-2^{m-\ell-1}}\\ &={2^{2\ell+2}\over 2^{2\ell+2}-1}(2^{m-2\ell-2}-1)(2^{m-2\ell-1}-1).
\end{align*}
So
\begin{align*}
B_\ell&={2^{2\ell+2}-1\over 2^{2\ell+2}}\cdot{1\over (2^{m-2\ell-2}-1)(2^{m-2\ell-1}-1)}\cdot B_{\ell+1}\\&\le {1\over (2^{m-2\ell-2}-1)(2^{m-2\ell-1}-1)}\cdot B_{\ell+1}\le {1\over 3}\cdot B_{\ell+1}.
\end{align*}
 \begin{align*}
 A_{2^{m-1}}&=2\left(2^m-1+\sum_{\ell=1}^{m-1\over 2}	B_\ell\right)\le 2\left(2^m-1+\sum_{\ell=1}^{m-1\over 2}	\left({1\over 3}\right)^{\ell-1}\cdot B_{m-1\over 2}\right)\\&\le 2\left(2^m-1+{3\over 2}\cdot B_{m-1\over 2}\right)=3\cdot A_{2^{m-1}-2^{m-{m-1\over 2}-1}}+2^{m+1}-2\\&\le 21\cdot {2^{m-1}\over 2^{m-1}-1}\cdot A_{2^{m-1}-2^{m-{m-1\over 2}-1}}.
 \end{align*}
The last inequality comes from that
\[A_{2^{m-1}-2^{m-{m-1\over 2}-1}}=(2^{{m-1\over 2}^2+{m-1\over 2}}){\prod_{i=2}^m(2^i-1)\over \prod_{i=1}^{m-1\over 2}(2^{2i}-1)}>2(2^m-1).\]
Secondly, we assume that $m$ is even. Then we have
\begin{equation*}
A_{2^{m-1}-2^{m-{m\over 2}-1}}/A_{2^{m-1}-2^{m-{m\over 2}}}=3\cdot {2^{m}\over 2^{m}-1}.
\end{equation*}
So we need to show that
\begin{equation*}
A_{2^{m-1}}\le 3\cdot {2^{m}\over 2^{m}-1}\cdot A_{2^{m-1}-2^{m-{m\over 2}-1}}.
\end{equation*}
Let $B_\ell=(2^{m-2\ell}-1)\cdot A_{2^{m-1}-2^{m-\ell-1}}$ for $1\le \ell \le {m\over 2}$. Then
\begin{align*}
 A_{2^{m-1}}&=2\left(2^m-1+\sum_{\ell=1}^{{m\over 2}-1}	B_\ell\right)\le 2\left(2^m-1+\sum_{\ell=1}^{{m\over 2}-1}	\left({1\over 3}\right)^{\ell-1}\cdot B_{{m\over 2}-1}\right)\\&\le 2\left(2^m-1+{3\over 2}\cdot B_{{m\over 2}-1}\right)\le 2\left(2^m-1+{1\over 2}\cdot B_{{m\over 2}}\right)\\&= A_{2^{m-1}-2^{m-{m\over 2}-1}}+2^{m+1}-2\le 3\cdot {2^{m-1}\over 2^{m-1}-1}\cdot A_{2^{m-1}-2^{m-{m-1\over 2}-1}}.
\end{align*}
This completes the proof.
\end{proof}

\begin{lem}(\cite{Li})
	The weight distribution of the second order homogeneous Reed-Muller code $HRM_q(2, m)$ over $\mathbb F_q$ is given as follows.
\begin{center}
\begin{tabular}{|c|c|}
 \hline Weight $w$ & Frequency $A_w$\\ \hline
	$0$ & $1$ \\[2ex] $q^{m}-q^{m-1}$ & $q^m-1+\sum_{\ell=1}^{\lfloor {m-1\over 2}\rfloor}(q^{\ell^2+\ell}){\prod_{i=m-2\ell}^m(q^i-1)\over \prod_{i=1}^\ell(q^{2i}-1)}$
	 \\[2ex]
	 \begin{tabular}{@{}c@{}}$q^{m}-q^{m-1}-\tau q^{m-j-1}$\\ ($1\le j\le \lfloor {m\over 2}\rfloor,\tau=\pm 1$)\end{tabular}
  	& ${q^{j^2}(q^j+\tau)\over 2}{\prod_{i=m-2j+1}^m(q^i-1)\over \prod_{i=1}^j(q^{2i}-1)}$\\\hline
\end{tabular}
\end{center}
and the weight distribution of the second order projective Reed-Muller code $PRM_q(2, m)$ over $\mathbb F_q$ is as follows.
\begin{center}
\begin{tabular}{|c|c|}
 \hline Weight $w$ & Frequency $A_w$\\ \hline
	$0$ & $1$ \\[2ex] $q^{m}$ & $q^{m+1}-1+\sum_{\ell=1}^{\lfloor {m\over 2}\rfloor}(q^{\ell^2+\ell}){\prod_{i=m-2\ell+1}^m(q^i-1)\over \prod_{i=1}^\ell(q^{2i}-1)}$\\[2ex]
	\begin{tabular}{@{}c@{}}$q^{m}-\tau q^{m-j}$\\ ($1\le j\le \lfloor {m+1\over 2}\rfloor, \tau=\pm 1$) \end{tabular} & ${q^{j^2}(q^j+\tau)\over 2}{\prod_{i=m-2j+2}^m(q^i-1)\over \prod_{i=1}^j(q^{2i}-1)}$\\ \hline
\end{tabular}
\end{center}
\end{lem}

\begin{lem}(\cite{Li}) \label{lem-RM}
There is a one-to-one correspondence between codewords of the second order projective Reed-Muller code $PRM_q(2, m)$	and codewords of the second order homogeneous Reed-Muller code $HRM_q(2, m+1)$.
\end{lem}

\begin{thm}
	The second order homogeneous Reed-Muller code $HRM_q(2, m)$ and the second order projective Reed-Muller code $PRM_q(2, m)$ are $1$-gap log-concave  or log-concave for $m\ge 2$.
\end{thm}
\begin{proof}
	Since there is a one-to-one corresponds between $HRM_q(2, m+1)$ and $PRM_q(2, m)$ by Lemma \ref{lem-RM}, we just consider the log-concavity for $HRM_q(2, m)$.\\
Let $A^q_{\tau, j}=A_{q^m-q^{m-1}-\tau q^{m-j-1}(q-1)}$. We first show that $A^q_{1, j}$ and $A^q_{-1, j}$ are both increasing sequences. For $2\le j\le \lfloor {m\over 2}\rfloor-1$, and $\tau\in\{1, -1\}$,
\begin{align} \label{eq-HRM}
\begin{split}
A^q_{\tau, j+1}/A^q_{\tau, j}&={q^{(j+1)^2}(q^{j+1}+\tau)\over q^{j^2}(q^{j}+\tau)}\cdot{\prod_{i=m-2j-1}^m(q^i-1)\prod_{i=1}^{j}(q^{2i}-1)\over \prod_{i=m-2j+1}^m(q^i-1)\prod_{i=1}^{j-1}(q^{2i}-1)}\\&=\left({q^{2j+1}\over q^{2j+2}-1}\right)\left({q^{j+1}+\tau\over q^{j}+\tau}\right)(q^{m-2j-1}-1)(q^{m-2j}-1)>1.
\end{split}
\end{align}
So the sequence of weight distribution of $HRM_q(2, m)$ is unimodal except $A_{q^m-q^{m-1}}$ which is the midpoint of the sequence.

Next, we show that $A^q_{\tau, j}$ is log-concave for $2\le j\le \lfloor {m\over 2}\rfloor-1$. Since
\[A^q_{\tau, j+1}={q^{(j+1)^2}(q^{j+1}+\tau)\over q^{j^2}(q^j+\tau)}\cdot{(q^{m-2j-1}-1)(q^{m-2j}-1)\over (q^{2j+2)}-1)}A^q_{\tau, j}\]
and
\[A^q_{\tau, j-1}={q^{(j-1)^2}(q^{j-1}+\tau)\over q^{j^2}(q^j+\tau)}\cdot{(q^{2j+2)}-1)\over (q^{m-2j+1}-1)(q^{m-2j+2}-1)}A^q_{\tau, j},\]
we have that
\begin{align*}
&A^q_{\tau, j+1}\cdot A^q_{\tau, j-1}\\&=q^2\cdot{(q^{j+1}+\tau)(q^{j-1}+\tau)\over (q^j+\tau)^2}\cdot{(q^{2j}-1)\over (q^{2j+2}-1)}\cdot{(q^{m-2j-1}-1)\over (q^{m-2j+1}-1)}\cdot{(q^{m-2j}-1)\over (q^{m-2j+2}-1)}\cdot (A^q_{\tau, j})^2\\
&\le {1\over q^4}\cdot{(q^{j+1}+\tau)(q^{j-1}+\tau)\over (q^j+\tau)^2}\cdot (A^q_{\tau, j})^2\le (A^q_{\tau, j})^2.
\end{align*}
So, $A^q_{\tau, j}$ is log-concave for $2\le j\le \lfloor {m\over 2}\rfloor-1$. For $j=1$,
\begin{align*}
A^q_{\tau, 2}&={q^4(q^2+\tau)\over 2}\cdot{(q^{m-3}-1)(q^{m-2}-1)(q^{m-1}-1)(q^{m}-1)\over (q^{2}-1)(q^{4}-1)}\\&\le 2q^2\cdot {q^2-1\over q^4-1}\cdot {(q^{m-3}-1)(q^{m-2}-1)\over (q^{m-1}-1)(q^m-1)}\cdot (A^q_{\tau, 1})^2\le (A^q_{\tau, 1})^2.
\end{align*}

Now we only need to consider $j=\lfloor{m\over 2}\rfloor$ case. We need to show that
\[A^q_{\tau, \lfloor{m\over 2}\rfloor}\cdot {A^q_{\tau, \lfloor{m\over 2}\rfloor}\over A^q_{\tau, \lfloor{m\over 2}\rfloor-1}}\ge A_{q^m-q^{m-1}}.\]
We consider two cases of $m$.

First, we assume that $m$ is odd. From Equation~\eqref{eq-HRM}, we have that
\[{A^q_{\tau, {m-1\over 2}}/A^q_{\tau, {m-1\over 2}-1}}={(q^{m-1})\over (q^{m-1}-1)}\cdot{(q^{(m-1)/2}+\tau)\over (q^{(m-1)/2}+ q\tau)}(q^2-1)(q^3-1).\]
Let $B_{\ell}=q^{\ell^2+\ell}{\prod_{i=m-2\ell}^m(q^i-1)\over \prod_{i=1}^\ell(q^{2i}-1)}$ for $1\le \ell\le \lfloor{m-1\over 2}\rfloor$.
Then
\[A_{q^m-q^{m-1}}=q^m-1+\sum_{\ell=1}^{m-1\over 2}B_\ell\]
and
\[B_\ell={2q^\ell\over q^\ell+\tau}(q^{m-2\ell}-1)A^q_{\tau, \ell}.\]
Since $B_{\ell+1}/B_\ell=q^{2\ell+2}{(q^{m-2\ell-2}-1)(q^{m-2\ell-1}-1)\over (q^{2\ell+2}-1)}$, we have that
\begin{align} \label{eq-HRM-2}
\begin{split}
B_{\ell+1}&={q^{2\ell+2}\over q^{2\ell+2}-1}(q^{m-2\ell-2}-1)(q^{m-2\ell-1}-1)B_\ell\\ &\ge 	(q^{m-2\ell-2}-1)(q^{m-2\ell-1}-1)B_\ell\\
&\ge (q-1)(q^2-1)B_\ell.
\end{split}
\end{align}
Then,
\begin{align*}
A_{q^m-q^{m-1}}&=q^m-1+\sum_{\ell=1}^{m-1\over 2}B_\ell	\le q^m-1+\sum_{\ell=1}^{m-1\over 2}\left({1\over (q-1)(q^2-1)}\right)^{\ell-1}B_{m-1\over 2}\\
&\le q^m-1+{(q-1)(q^2-1)\over (q-1)(q^2-1)-1}B_{m-1\over 2}\\
&=q^m-1+2{(q-1)(q^2-1)\over (q-1)(q^2-1)-1}\cdot{2(q-1)q^{(m-1)/2}\over q^{(m-1)/2}+\tau}A^q_{\tau, {m-1\over 2}}\\
&\le {(q^{m-1})\over (q^{m-1}-1)}{(q^{(m-1)/2}+\tau)\over (q^{(m-1)/2}+ q\tau)}(q^2-1)(q^3-1)A^q_{\tau, {m-1\over 2}}\\
&=\left({A^q_{\tau, {m-1\over 2}}/A^q_{\tau, {m-1\over 2}-1}}\right)A^q_{\tau, {m-1\over 2}}.
\end{align*}
The last inequality comes from the fact that when $q=2$, ${(q-1)(q^2-1)\over (q-1)(q^2-1)-1}\cdot{2q^{(m-1)/2}\over q^{(m-1)/2}+\tau}$ has the largest value whereas ${(q^{m-1})\over (q^{m-1}-1)}{(q^{(m-1)/2}+\tau)\over (q^{(m-1)/2}+ q\tau)}(q+1)(q^3-1)$ has the smallest value, and even in such a case, we have
\[ 3{2^{(m-1)/2}\over 2^{(m-1)/2}+\tau}<21{(2^{m-1})\over (2^{m-1}-1)}{(2^{(m-1)/2}+\tau)\over (2^{(m-1)/2}+ 2\tau)},
\]
and
\[ A^q_{\tau, {m-1\over 2}} = {q^{{m-1\over 2}^2}(q^{m-1\over 2}+\tau)\over 2}{\prod_{i=2}^m(q^i-1)\over \prod_{i=1}^{m-1\over 2}(q^{2i}-1)} > q^m-1.\]

Now we assume that $m$ is even. Note that
\begin{equation*}
	{A^q_{\tau, {m\over 2}}/A^q_{\tau, {m\over 2}-1}}={q^{m}\over q^{m}-1}\cdot{(q^{m/2}+\tau)\over (q^{m/2}+q\tau)}(q-1)(q^2-1).
\end{equation*}
Since
\[A_{q^m-q^{m-1}}=q^m-1+\sum_{\ell=1}^{{m\over 2}-1}B_\ell,\]
and
\[B_{\ell+1}\ge (q^{m-2\ell-2}-1)(q^{m-2\ell-1}-1)B_\ell\ge (q^2-1)(q^3-1)B_\ell\]
by Equation \eqref{eq-HRM-2}, we have that
\begin{align*}
A_{q^m-q^{m-1}}&=q^m-1+\sum_{\ell=1}^{{m\over 2}-1}B_\ell	\le q^m-1+\sum_{\ell=1}^{{m\over 2}-1}\left({1\over (q^2-1)(q^3-1)}\right)^{\ell-1}B_{{m\over 2}-1}\\
&\le q^m-1+{(q^2-1)(q^3-1)\over (q^2-1)(q^3-1)-1}B_{{m\over 2}-1}\\
&\le q^m-1+2{(q^2-1)(q^3-1)\over (q^2-1)(q^3-1)-1}\cdot{q^{{m\over2}-1}\over q^{{m\over2}-1}+\tau}(q^2-1)A^q_{\tau, {m\over 2}-1}\\
&=q^m-1+2{(q^2-1)(q^3-1)\over (q^2-1)(q^3-1)-1}\cdot{q^{{m\over2}}\over q^{{m\over2}}+\tau}\cdot{q^m-1\over q^m}\cdot{1\over q-1}A^q_{\tau, {m\over 2}}.
\end{align*}
Our next calculation depends on $q$.
First, suppose that $q\ge 3$. Then we have the following. Since
\[{(q^2-1)(q^3-1)\over (q^2-1)(q^3-1)-1}, {q^{{m\over2}}\over q^{{m\over2}}+\tau}<2,
\]
\[
{q^m-1\over q^m},  {2\over q-1}\le 1,\quad\mbox{ and }\quad
q^m-1<A^q_{\tau, {m\over 2}},
\]
we have that
\begin{align*}
A_{q^m-q^{m-1}} &\le 5A^q_{\tau, {m\over 2}}\\&\le 16\cdot{q^{m}\over q^{m}-1}\cdot{(q^{m/2}+\tau)\over (q^{m/2}+q\tau)}A^q_{\tau, {m\over 2}}\\&\le {q^{m}\over q^{m}-1}\cdot{(q^{m/2}+\tau)\over (q^{m/2}+q\tau)}(q-1)(q^2-1){A^q_{\tau, {m\over 2}}}.
\end{align*}
Next we suppose $q=2$. Then we have the following.
\begin{align*}
A_{2^m-2^{m-1}}&\le 2^m-1+{21\over 10}\cdot{2^{{m\over2}}\over 2^{{m\over2}}+\tau}\cdot{2^m-1\over 2^m}A^2_{\tau, {m\over 2}}\\&\le 3\cdot{2^{m}\over 2^{m}-1}\cdot{(2^{m/2}+\tau)\over (2^{m/2}+2\tau)}A^2_{\tau, {m\over 2}}.
\end{align*}
So the weight distribution of $HRM_q(2, m)$ is log-concave except for $A_{q^m-q^{m-1}}$.

Since there is a one-to-one corresponds between $HRM_q(2, m+1)$ and $PRM_q(2, m)$ by Lemma \ref{lem-RM}, the weight distribution of $PRM_q(2, m)$ is also $1$-gap log-concave or log-concave.
\end{proof}

\begin{prop}
	The second order homogeneous Reed-Muller code $HRM_q(2, m)$ is log-concave for $m\ge 2$ if $m$ is odd.
\end{prop}

\begin{proof}
We define $B_\ell=q^{\ell^2+\ell}{\prod_{i=m-2\ell}^m(q^i-1)\over \prod_{i=1}^j(q^{2i}-1)}$ as in the proof of Theorem 9. Since
	\[A^q_{\tau, {m-1\over 2}}\le 2\cdot {q^{m-1\over 2}\over q^{m-1\over 2}+\tau}(q-1)B_{{m-1\over 2}}<A_{q^m-q^{m-1}},\]
the weight distribution of $HRM_q(2, m)$ is unimodal. So we have the desired result.
\end{proof}

\begin{prop}
	The second order projective Reed-Muller code $PRM_q(2, m)$ is log-concave for $m\ge 2$ if $m$ is even.
\end{prop}

\begin{proof}
It follows from the fact that there is a one-to-one correspondence between codewords of $PRM_q(2, m)$ and codewords of $HRM_q(2, m+1)$.
\end{proof}


\medskip


\section{MDS codes}

It is known that the weight enumerator of an MDS code is unique (see, e.g., \cite[12, Ch.11, \S 3, Theorem 6]{MacSlo}), which is given by $A_0 = 1$, $A_w = 0$ for $0 < w < d$, and
\begin{equation} \label{eq:MDS-weight}
	A_w = \binom{n}{w} (q-1) \sum_{j=0}^{w-d} (-1)^j \binom{w-1}{j} q^{w-d-j}.
\end{equation}
According to \cite{MDS-weight-number}, many $[n,k,d]_q$ MDS codes have $k$ nonzero weights, with the exceptions as follows:
\begin{enumerate}

	\item[{(i)}] If $k=2$ and $n = q+1$, then such MDS code is a Simplex code, which contains only words of weight zero or $q$.
	
	\item[{(ii)}] If $n = q+2$, $q = 2^m$, and $k = 3$, then $A_{2^m} = (2^{2m}-1)(2^{m-1}+1)$ and $A_n = 2^{m-1}(2^m-1)^2$.
\end{enumerate}

As for the MDS codes with $2$ nonzero weights, we have
\begin{enumerate}

	\item[{(i)}] If $k = 2$ and $n \leqslant q$, then $A_{n-1} = (q-1)n$, $A_n = (q-1)(q-n+1)$, and
	\[ \frac{A_{n-1}^2}{A_0 A_n} = \frac{(q-1)n^2}{q-n+1} > \frac{(q-1)n^2}{q-1} > 1. \]
	In fact,
	\[ \frac{A_d^2}{A_0 A_{d+1}} = \frac{(q-1)\binom{n}{k-1}^2}{\binom{n}{k-2} (q-n+k-1)} > 1 \]
	for all MDS codes.
	\item[{(ii)}] If $n = q+2$, $q = 2^m$, and $k = 3$, then then it is a two-weight code, and
	\[ A_{2^m} = (2^{2m}-1)(2^{m-1}+1) = (2^m-1)(2^m+1) (2^{m-1}+1) > 2^{m-1}(2^m-1)^2 = A_n. \]
\end{enumerate}

In this subsection, we will consider the log-concave MDS codes.
From \eqref{eq:MDS-weight}, we have
\[
	\frac{A_{w}^2}{A_{w-1} A_{w+1}} = \frac{w+1}{w} \cdot \frac{n-w+1}{n-w} \cdot  \frac{\left(\sum_{j=0}^{s} (-1)^j \binom{w-1}{j} q^{s-j} \right)^2}{\left(\sum_{j=0}^{s-1} (-1)^j \binom{w-2}{j} q^{s-1-j} \right) \left(\sum_{j=0}^{s+1} (-1)^j \binom{w}{j} q^{s+1-j} \right)}.
\]
Let $s =  w - d = w - n + k -1$, $k \geqslant 3$, $q \geqslant 2$ be positive integers.
Then we define
\[ f_k(s,w,q) = \sum_{j=0}^{s} (-1)^j \binom{w-1}{j} q^{-j}, \]
and
\begin{equation} \label{eq: ratio-MDS}
	G_k(s,n,q) = \frac{A_{w}^2}{A_{w-1} A_{w+1}} = \frac{w+1}{w} \cdot \frac{n-w+1}{n-w} \cdot  \frac{f_k^2(s,w,q)}{f_k(s-1,w-1,q) f_k(s+1,w+1,q)}.
\end{equation}

\begin{lem}
	About $f_k(s,w,q)$, we have
	\begin{enumerate}
		\renewcommand{\labelenumi}{(\theenumi)}
		\item[{(i)}] $f_k(s,w,q) > 0$ if $0 \leqslant s \leqslant k-1$, $w \leqslant q+1$, with the exception $f_2(1,q+1,q) = 0$;
		\item[{(ii)}] $f_k(s+1, w+1, q) = \left(1-\frac{1}{q}\right) f_k(s,w,q) + (-1)^{s+1} \binom{w-1}{s+1} q^{-s-1}$;
		\item[{(iii)}] If $s$ is even, then $f_k(s+1,w+1,q) < \left(1-\frac{1}{q}\right) f_k(s,w,q)$;
		\item[{(iv)}] If $s$ is odd, then $f_k(s+1,w+1,q) > \left(1-\frac{1}{q}\right) f_k(s,w,q)$;
	\end{enumerate}
\end{lem}
\begin{proof}
	It is clear that $f_k(0,w,q) = 1$, and $f_k(1,w,q) = 1 - (w-1)/q$.
	Obviously, $f_k(1,w,q) > 0$ if $w < q+1$ and $f_k(1,w,q) = 0$ if $w = q+1 \geqslant n$.
	
	Let $2 \leqslant j \leqslant s$ be even. Then
	\[ \frac{q \binom{w-1}{j}}{\binom{w-1}{j+1}} = q \frac{j+1}{w-j-1} > 1 \]
	if and only if $(q+1)j > w-1-q$, which holds when $w \leqslant n \leqslant q+1$.
	That is to say,  \[ (-1)^j \binom{w-1}{j} q^{-j} + (-1)^{j+1} \binom{w-1}{j+1} q^{-j-1} > 0. \]
	Thus, $f_k(s,w,q) > 0$ when $s \geqslant 2$.
	It is not difficult to check that
	\begin{align*}
		f_k(s+1, w+1, q) & = 1 + \sum_{j=1}^{s+1} (-1)^j \left(\binom{w-1}{j} + \binom{w-1}{j-1} \right) q^{-j}
		\\
		& = 1 + \sum_{j=1}^{s+1} (-1)^j \binom{w-1}{j} q^{-j} + \sum_{j=1}^{s+1} (-1)^j \binom{w-1}{j-1} q^{-j} \\
		& = f_k(s+1,w,q) +  \sum_{i=0}^{s} (-1)^{i+1} \binom{w-1}{i} q^{-i-1} \\
		& = f_k(s+1,w,q) - \frac{1}{q} f_k(s,w,q) \\
		& = \left(1-\frac{1}{q}\right) f_k(s,w,q) + (-1)^{s+1} \binom{w-1}{s+1} q^{-s-1}.
	\end{align*}
	Moreover, it is easy to obtain $(3)$ and $(4)$ from the above formula.
\end{proof}

Next, consider the following function: \[ g_k(s,w,q) = \frac{f_k^2(s,w,q)}{f_k(s-1,w-1,q) f_k(s+1,w+1,q)}.  \]

\begin{lem} \label{lemma: MDS-g}
	About $g_k(s,w,q)$, we have
	\begin{enumerate} \renewcommand{\labelenumi}{(\theenumi)}
		\item[{(i)}] If $s$ is even, then $g_k(s,w,q) > 1$;
		\item[{(ii)}] If $s$ is odd, then $g_k(s,w+1,q) / g_k(s,w,q) < 1$;
		\item[{(iii)}] If $s$ is odd, then $g_k(s+2,w+2,q) / g_k(s,w,q) > 1$.
	\end{enumerate}
\end{lem}
\begin{proof}
\begin{enumerate}
		\item[{(i)}] If $s$ is even, then
		\begin{align*}
			\frac{f_k^2(s,w,q)}{f_k(s-1,w-1,q) f_k(s+1,w+1,q)} & > \frac{f_k(s,w,q)}{\left(1-\frac{1}{q}\right)f_k(s-1,w-1,q)} > 1.
		\end{align*}
		
		\item[{(ii)}] If $s$ is odd,
		then
		\begin{align*}
			\frac{f_k(s,w+1,q)}{f_k(s,w,q)} & = 1 - \frac{f_k(s-1,w,q)}{q f_k(s,w,q)} < 1 - \frac{1}{q}, \\
			\frac{f_k(s-1,w,q)}{f_k(s-1,w-1,q)} & = 1 - \frac{f_k(s-2,w-1,q)}{q f_k(s-1,w-1,q)} > 1 - \frac{1}{q}.
		\end{align*}
		Thus,
		\begin{align*}
			\frac{g_k(s,w+1,q)}{g_k(s,w,q)} & = \frac{f_k^2(s,w+1,q)}{f_k^2(s,w,q)} \cdot  \frac{f_k(s-1,w-1,q)}{f_k(s-1,w,q)} \cdot \frac{f_k(s+1,w+1,q)}{f_k(s+1,w+2,q)} < 1.
		\end{align*}
		
		\item[{(iii)}] 	It is easy to check that
		\begin{align*}
		&f_k(s+2,w+2,q) \\ &= \frac{(q-1)^2}{q^2} f_k(s,w,q)  + (-1)^{s+1} \frac{q-1}{q}\binom{w-1}{s+1} q^{-s-1} + (-1)^{s+2} \binom{w}{s+2} q^{-s-2},
		\end{align*}
		where $w = n-k+1+s$.
		According to the definition of $g_k(s+2,w+2,q)$, we have $w \leqslant n-3$.
		Since $n \leqslant q+2$, we have
		\[ (q-1) \binom{w-1}{s+1} / \binom{w}{s+2} = \frac{(q-1)(s+2)}{w} \geqslant \frac{(n-3)(s+2)}{n-3} > 1. \]
		Thus, when $s$ is odd,
		\begin{align*}
			\frac{f_k(s+2,w+2,q)}{f_k(s,w,q)} > \left(1 - \frac{1}{q}\right)^2, \textnormal{and} \
			\frac{f_k(s-1,w-1,q)}{f_k(s+1,w+1,q)} > \frac{1}{\left(1 - \frac{1}{q}\right)^2}.
		\end{align*}
		Hence, when $s$ is odd,
		\begin{align*}
			\frac{g_k(s+2,w+2,q)}{g_k(s,w,q)} & = \frac{f_k^2(s+2,w+2,q)}{f_k^2(s,w,q)} \cdot \frac{f_k(s-1,w-1,q)}{f_k(s+1,w+1,q)} \cdot  \frac{f_k(s+1,w+1,q)}{f_k(s+3,w+3,q)} > 1.
		\end{align*}
	\end{enumerate}
	We conclude the proof.
\end{proof}

\begin{thm}
	If $g_k(s,k+2,q)>1$ for all possible odd $s$, then there exists an $n_0(k,q)$ such that for all $n < n_0(k,q)$, the $[n,k,d]_q$ MDS code (if it exists) is log-concave property.
\end{thm}
\begin{proof}
	Let $d < w < n$, $w-d = s$.
	Then
	\[ \frac{A_{w}^2}{A_{w-1} A_{w+1}} = G_k(s,n,q) = \frac{n-k+2+s}{n-k+1+s} \cdot  \frac{k-s}{k-1-s} \cdot g_k(s,w,q). \]
	If $s$ is even, then $g_k(s,w,q) > 1$, and $A_w^2 > A_{w-1} A_{w+1}$.
	Note that
	\begin{align} \label{eq: MDS-relation}
		\frac{G_k(s,n+1,q)}{G_k(s,n,q)} & = \frac{n-k+3+s}{n-k+2+s} \cdot \frac{n-k+1+s}{n-k+2+s} \cdot \frac{g_k(s,w+1,q)}{g_k(s,w,q)} \notag \\
		& = \frac{T^2-1}{T^2} \cdot \frac{g_k(s,w+1,q)}{g_k(s,w,q)} < 1,
	\end{align}
	which implies that $G_k(s,n,q)$ decreases when $n$ becomes larger, where $T = n-k+2+s$.
	Therefore, if $g_k(s,k+2,q)>1$ for all possible odd $s$, then there exists an $n_0(k,q)$ such that for all $n < n_0(k,q)$, the $[n,k,d]_q$ MDS code (if it exists) satisfies the log-concave property.
\end{proof}

\begin{thm} \label{thm: MDS-log}
	If $C$ is an $[n,k,d]_q$ MDS code with $3 \leqslant k \leqslant n/2 + 3$, then $C$ is log-concave when $q \geqslant q_0(n,k)$ and $C$ is not log-concave when $q < q_0(n,k)$, where $q_0(n,k)$ is the larger root of the quadratic polynomial
	\[
		(m+k-1) q^2 - (km^2 - 2k +2) q + \frac{km^3-km^2}{2} -km + m+k-1,
	\]
	and $m = n-k+2$.
	Moreover, if $C$ is log-concave, then each $[n', k, d']_q$ MDS code is log-concave, where $n' < n$ and $k\leqslant n'/2 + 3$ may not hold.
\end{thm}
\begin{proof}
	We only need to consider odd $s$.
	According to Lemma \ref{lemma: MDS-g}, $g_k(s+2,w+2,q) > g_k(s,w,q)$ if $1 \leqslant s \leqslant k-4$ is odd.
	It is clear that
	\[ \frac{A^2_{w+2}}{A_{w+1} A_{w+3}} / \frac{A^2_{w}}{A_{w-1} A_{w+1}} = \frac{(n-t+4)(n-t+1)}{(n-t+3)(n-t+2)} \cdot \frac{t^2-3t+2}{t^2-3t} \cdot \frac{g_k(s+2,w+2,q)}{g_k(s,w,q)} >1, \]
	where $t = k-s \geqslant 1$, if
	\begin{align*}
		& \frac{(n-t+4)(n-t+1)}{(n-t+3)(n-t+2)} \cdot \frac{t^2-3t+2}{t^2-3t} \\
		= & \left( 1 - \frac{2}{(n-t+2)(n-t+3)} \right) \left( 1 + \frac{2}{t^2-3t} \right) \geqslant 1,
	\end{align*}
	which is equivalent to
	\[ (t-1)(t-2) = (t^2-3t) + 2 \leqslant (n-t+2)(n-t+3). \]
	It is from the fact that \[ \left(1-\frac{2}{a} \right) \left(1+ \frac{2}{b}\right) \geqslant 1 \Longleftrightarrow b+2 \leqslant a, \textnormal{where} \  a,b>0. \]
	That is to show \[ \frac{(n-t+2)(n-t+3)}{(t-1)(t-2)} \geqslant 1, \]
	which holds if $n-t+2 \geqslant t-2$, i.e., $s \geqslant k - 2 - n/2$.
	If $k \leqslant n/2 + 3$, then
	\[ \frac{A^2_{d+1}}{A_{d} A_{d+2}} < \frac{A^2_{d+3}}{A_{d+2} A_{d+4}} < \frac{A^2_{d+5}}{A_{d+4} A_{d+6}} < \cdots . \]
	For $s=1$, we have
	\[ \frac{A^2_{d+1}}{A_{d} A_{d+2}} = G_k(1,n,q) = \frac{n-k+3}{n-k+2} \cdot \frac{k-1}{k-2} \cdot \frac{(1-(n-k+1)/q)^2}{1- (n-k+2)/q + \binom{n-k+2}{2}/q^2} \geqslant 1 \]
	if and only if
	\[ \frac{m+1}{m} \cdot \frac{k-1}{k-2} \cdot \frac{(q-m+1)^2}{q^2-mq+\binom{m}{2}} \geqslant 1, \  \textnormal{where} \ m = n-k+2 \geqslant k-4, \]
	which is equivalent to show
	\begin{equation}\label{eq: MDS-NC}
		(m+k-1) q^2 - (km^2 - 2k +2) q + \frac{km^3-km^2}{2} -km + m+k-1 \geqslant 0.
	\end{equation}
	Note that $m>2$ and the discriminant is \[ m(k - 2)(m - 2)(km^2 - 2k + 2) > 0, \]
	then $q \geqslant q_0(n,k)$, the large root of the quadratic polynomial.
	In fact, about the smaller root $q_0'(n,k)$, we have
	\begin{align*}
		q_0'(n,k) & = \frac{(km^2- 2k +2) - \sqrt{m(k - 2)(m - 2)(km^2 - 2k + 2)}}{2(m+k-1)} \\
		& = \frac{\sqrt{km^2- 2k +2}}{2(m+k-1)} \left( \sqrt{km^2- 2k +2} - \sqrt{m(k - 2)(m - 2)} \right) \\
		& = \frac{\sqrt{km^2- 2k +2}}{2(m+k-1)} \cdot \frac{2m^2+2mk-4m-2k+2}{\sqrt{km^2- 2k +2} + \sqrt{m(k - 2)(m - 2)}} \\
		& < \frac{\sqrt{km^2- 2k +2}}{2(m+k-1)} \cdot \frac{2m^2+2mk-4m-2k+2}{2\sqrt{m(k - 2)(m - 2)}} \\
		& < \frac{\sqrt{km^2- 2k +2}}{2(m+k-1)} \cdot \frac{m(k+m-2)}{\sqrt{m(k - 2)(m - 2)}}
		  < \frac{m(k+m-2)}{2(m+k-1)} \cdot \frac{m \sqrt{k}}{(m - 2)\sqrt{k - 2}} \\
		& < \frac{m}{2} \cdot \frac{m \sqrt{k}}{(m - 2)\sqrt{k - 2}} < m, \textnormal{when} \ m \geqslant 7 \ \textnormal{and} \ k \geqslant 4.
	\end{align*}
	If $m = 3$, then $n - k = 1$, the log-concavity for such code is clear.
	If $k = 3$, then
	\[ q_0'(n,3) = \frac{(3m^2-4)-\sqrt{m(m-2)(3m^2-4)}}{2(m+2)} < \frac{(3m^2-4)-\sqrt{3}(m-2)(m-1)}{2(m+2)} < m \]
	for all $m \geqslant 3$.
	If $k \geqslant 4$ and $4 \leqslant m \leqslant 6$, then
	\[ q_0'(n,k) < \frac{m^2 \sqrt{k}}{2(m - 2)\sqrt{k - 2}} \leqslant \frac{m^2}{\sqrt{2}(m-2)} \leqslant \frac{9}{\sqrt{2}} < 7. \]
	There are few MDS codes over $\mathbb{F}_3$, $\mathbb{F}_4$, $\mathbb{F}_5$, and none of them satisfy $k \geqslant 4$ and $m \geqslant 5$.
	As for $m = 4$, the only $[6,4,3]_5$ MDS code is a Hamming code, and it is not log-concave.
	
	That is to say, if $q < q'_0(n,k)$, then such a code is either not log-concave, or does not exist by the MDS conjecture since $q_0'(n,k) < m < n-1$ for many $k$ and $m$.
	Meanwhile, $q_0(n,k) > km^2- 2k +2 -m$.
	
	Besides, if $C'$ is an $[n',k,d']_q$ MDS code, then by Equation~\eqref{eq: MDS-relation}, $C'$ is log-concave, where $k \leqslant n'/2+3$ may not hold.
\end{proof}

\begin{rem}{\em
	It is clear that $q \geqslant q_0(n,k)$ is necessary for every MDS code to be log-concave.
	And the above theorem shows that $q \geqslant q_0(n,k)$ is sufficient when $k \leqslant n/2+3$.
	As for the log-concavity when $k > n/2 + 3$, one method is to check the log-concavity of the longer MDS code (with the same dimension), and another is to give a similar condition as $q \geqslant q_0(n,k)$.
	For example, if $k = n/2+4$ or $k=n/2+5$, the condition for $G_k(3,n,q) \geqslant 1$ is necessary.
	
	Usually, if $n/2+2s \leqslant k \leqslant n/2 + 2s+1$, then the condition for $G_k(s,n,q) \geqslant 1$ is necessary.
	Considering that $g_k(s+2,w+2,q) > g_k(s,w,q)$ for odd $s$, these conditions can be simplified.
	When we solve $G_k(s,n,q) \geqslant 1$, we just replace $g_k(s,w,q)$ with $g_k(1,d+1,q)$, then we can obtain similar quadratic inequality as in Equation~\eqref{eq: MDS-NC}, and a similar condition $q \geqslant q_s(n,k)$, which is sufficient, but maybe not necessary.
	In fact, this technique has been applied in the proof of Theorem \ref{thm-MDS-Hamming} and works well.
}
\end{rem}

\begin{ex}{\em
	According to \cite[Theorem 5.3.4]{HuffmanPless}, there exists an $[n,k,d]_q$ MDS code, when $1 \leqslant k \leqslant n \leqslant q+1$.
	Let $k=3$ and $n=5$, then the MDS codes exist when $q \geqslant 4$.
	And $G_k(1,5,q) = 6q^2 - 44q + 66 \geqslant 1$ if $q \geqslant q_0(5,3)$ where $5 < q_0(5,3) < 6$.
	That is to say, the $[5,3,3]_q$ MDS code is log-concave when $q \geqslant 7$.
	In fact, the weight distribution of the $[5,3,3]_5$ code is
	\[ A_0 = 1, A_3 = 40, A_4 = 40, A_5 = 44. \]
	And the weight distribution of the $[5,3,3]_4$ code is
	\[ A_0 = 1, A_3 = 30, A_4 = 15, A_5 = 18. \]
}
\end{ex}

\begin{ex}{\em
	Let $k = 9$ and $n = 12$. Then $G_9(1,12,q) = 13q^2 - 209q + 418$, so $ 13< q_0(n,k) < 14$. Thus, the $[12, 9, 4]_q$ code for each $q \geqslant 16$ is log-concave.
	Moreover, the $[11, 9, 3]_q$ code for each $q \geqslant 16$ is also log-concave while $9 > (11/2)+3 = 8.5$.
}
\end{ex}

\section{Conclusion}
For the first time, we have studied the nonzero weight distribution of a linear code from the standpoint of log-concavity. In particular, we have shown that all binary Hamming codes of length $2^r -1$ ($r=3$ or $r \ge 5$), the extended Hamming codes of length $2^r ~(r \ge 3)$, and the second order Reed-Muller codes $R(2, m)~ (m \ge 2),$ as well as the homogeneous and projective second order Reed-Muller codes are log-concave. We have also shown that MDS codes of moderate lengths and dimensions are log-concave.

In future works, it will be interesting to investigate the log-concavity of classical families of linear codes such as few-weight codes, cyclic or quasi-cyclic codes, quadratic residue codes, double circulant codes, Griesmer codes, AG codes, and self-dual codes. Other integer sequences attached to codes, like the coset leader weight distribution, the  weight distributions of cosets, or the generalized Hamming weight distribution are also worth considering.

{\bf Acknowledgment:} The authors would like to thank Patrick Sol\'{e} for proposing this topic to us, and suggesting valuable comments.

\end{document}